\documentclass[a4paper,british,fleqn]{article}
\usepackage[T1]{fontenc}
\usepackage[latin9]{inputenc}
\usepackage{xcolor}
\usepackage{pdfcolmk}
\usepackage{pifont}
\usepackage{float}
\usepackage{url}
\usepackage{amsmath}
\usepackage{amsthm}
\usepackage{amssymb}
\usepackage{graphicx}
\usepackage[authoryear]{natbib}
\PassOptionsToPackage{normalem}{ulem}
\usepackage{ulem}

\makeatletter

\pdfpageheight\paperheight
\pdfpagewidth\paperwidth

\newcommand{\lyxmathsym}[1]{\ifmmode\begingroup\def\b@ld{bold}
  \text{\ifx\math@version\b@ld\bfseries\fi#1}\endgroup\else#1\fi}

\floatstyle{ruled}
\newfloat{algorithm}{tbp}{loa}
\providecommand{\algorithmname}{Algorithm}
\floatname{algorithm}{\protect\algorithmname}
\providecolor{lyxadded}{rgb}{0,0,1}
\providecolor{lyxdeleted}{rgb}{1,0,0}

\DeclareRobustCommand{\lyxsout}[1]{\ifx\\#1\else\sout{#1}\fi}

\theoremstyle{plain}
\newtheorem{thm}{\protect\theoremname}
  \theoremstyle{plain}
  \newtheorem{prop}[thm]{\protect\propositionname}
  \theoremstyle{plain}
  \newtheorem{cor}[thm]{\protect\corollaryname}

\usepackage{fullpage}
\usepackage{textcomp}
\usepackage{txfonts}
\usepackage{times}
\usepackage{setspace}
\usepackage{color}

\setcitestyle{round}

\date{}

\makeatother

\usepackage{babel}
  \providecommand{\corollaryname}{Corollary}
  \providecommand{\propositionname}{Proposition}
\providecommand{\theoremname}{Theorem}

\begin{document}

\title{Anytime Monte Carlo}

\author{Lawrence M. Murray\\
Uber AI  \and Sumeetpal Singh\\
University of Cambridge\\
\& The Alan Turing Institute \and Anthony Lee\\
University of Bristol\\
\& The Alan Turing Institute}
\maketitle
\begin{abstract}
Monte Carlo algorithms simulate some prescribed number
of samples, taking some random real time to complete the computations
necessary. This work considers the converse: to impose a real-time
budget on the computation, which results in the number of samples simulated
being random. To complicate matters, the real time taken for each simulation
may depend on the sample produced, so that the samples themselves
are not independent of their number, and a length bias with respect
to compute time is apparent. This is especially problematic when a
Markov chain Monte Carlo (MCMC) algorithm is used and the final state
of the Markov chain\textemdash rather than an average over all states\textemdash is
required, which is the case in parallel tempering implementations of MCMC. The length bias does not diminish with the compute budget
in this case. It also occurs in sequential Monte Carlo (SMC)
algorithms, which is the focus of this paper. We propose an \emph{anytime} framework to address the
concern, using a continuous-time Markov jump process to study the
progress of the computation in real time. We first show that for any MCMC algorithm, the length
bias of the final state's distribution due to the imposed real-time computing budget  can be eliminated by using a multiple
chain construction. The utility of this construction is then demonstrated
on a large-scale SMC$^{2}$ implementation, using four billion particles
distributed across a cluster of 128 graphics processing units on the
Amazon EC2 service. The anytime framework imposes a real-time budget
on the MCMC move steps within the SMC$^{2}$ algorithm, ensuring that all processors
are simultaneously ready for the resampling step, demonstrably reducing
idleness to due waiting times and providing substantial control over the total compute
budget.
\end{abstract}

\section{Introduction}

Real-time budgets arise in embedded systems, fault-tolerant computing,
energy-constrained computing, distributed computing and, potentially,
management of cloud computing expenses and the fair computational
comparison of methods. Here, we are particularly interested in the
development of Monte Carlo algorithms to observe such budgets, as
well as the statistical properties\textemdash and limitations\textemdash of
these algorithms. While the approach has broader applications, the
pressing motivation in this work is the deployment of Monte Carlo
methods on large-scale distributed computing systems, using real-time
budgets to ensure the simultaneous readiness of multiple processors
for collective communication, minimising the idle wait time that typically
precedes it.

A conceptual solution is found in the \emph{anytime algorithm}. An
anytime algorithm maintains a numerical result at all times, and will
improve upon this result if afforded extra time. When interrupted,
it can always return the current result. Consider, for example, a
greedy optimisation algorithm: its initial result is little more than
a guess, but at each step it improves upon that result according to
some objective function. If interrupted, it may not return the optimal
result, but it should have improved on the initial guess. A conventional
Markov chain Monte Carlo (MCMC) algorithm, however, is not anytime
if we are interested in the final state of the Markov chain at some
interruption time: as will be shown, the distribution of this state
is length-biased by compute time, and this bias does not reduce when
the algorithm is afforded additional time.

MCMC algorithms are typically run for some long time and, after removing
an initial burn-in period, the expectations of some functionals of
interest are estimated from the remainder of the chain. The prescription
of a real-time budget, $t$, introduces an additional bias in these
estimates, as the number of simulations that will have been completed
is a random variable, $N(t)$. This bias diminishes in $t$, and for
long-run Markov chains may be rendered negligible. In motivating this
work, however, we have in mind situations where the final state of
the chain is most important, rather than averages over all states.
The bias in the final state does not diminish in $t$. Examples where
this may be important include (a) sequential Monte Carlo (SMC), where,
after resampling, a small number of local MCMC moves are performed
on each particle before the next resampling step, and (b) parallel
tempering, where, after swapping between chains, a number of local
MCMC moves are performed on each chain before the next swap. In a
distributed computing setting, the resampling step of SMC, and the
swap step of parallel tempering, require the synchronisation of multiple
processes, such that all processors must idle until the slowest completes.
By fixing a real-time budget for local MCMC moves, we can reduce this
idle time and eliminate the bottleneck, but must ensure that the length
bias imposed by the real-time budget is negligible, if not eliminated
entirely. In a companion paper \citep{ASM20}, we also develop this idea for parallel tempering based MCMC. 

The compute time of a Markov chain depends on exogenous factors such
as processor hardware, memory bandwidth, I/O load, network traffic,
and other jobs contesting the same processor. But, importantly, it
may also depend on the states of the Markov chain. Consider (a) inference
for a mixture model where one parameter gives the number of components,
and where the time taken to evaluate the likelihood of a data set
is proportional to the number of components; (b) a differential model
that is numerically integrated forward in time with an adaptive time
step, where the number of steps required across any given interval
is influenced by parameters; (c) a complex posterior distribution
simulated using a pseudomarginal method~\citep{Andrieu2009}, where
the number of samples required to marginalise latent variables depends
on the parameters; (d) a model requiring approximate Bayesian computation
(ABC) with a rejection sampling mechanism, where the acceptance rate
is higher for parameters with higher likelihood, and so the compute
time lower, and vice versa. Even in simple cases there may be a hidden
dependency. Consider, for example, a Metropolis\textendash Hastings
sampler where there is some seemingly innocuous book-keeping code,
such as for output, to be run when a proposal is accepted, but not
when a proposal is rejected. States from which the proposal is more
likely to accept then have longer expected hold time due to this book-keeping
code. In general, we should assume that there is indeed a dependency,
and assess whether the resulting bias is appreciable.

The first major contribution of the present work is a framework for
converting any existing Monte Carlo algorithm into an anytime Monte
Carlo algorithm, by running multiple Markov chains in a particular
manner. The framework can be applied in numerous contexts where real-time
considerations might be beneficial. The second major contribution
is an application in one such context: an SMC$^{2}$ algorithm deployed
on a large-scale distributed computing system. An anytime treatment
is applied to the MCMC moves steps that precede resampling, which
requires synchronisation between processors, and can be a bottleneck
in distributed deployment of the algorithm. The real-time budget reduces
wait time at synchronisation, relieves the resampling bottleneck,
provides direct control over the compute budget for the most expensive
part of the computation, and in doing so provides indirect control
over the total compute budget.

Anytime Monte Carlo algorithms have recently garnered some interest.
\citet{Paige2014} propose an anytime SMC algorithm called the \emph{particle
cascade}. This transforms the structure of conventional SMC by running
particles to completion one by one, with the sufficient statistics
of preceding particles used to make birth-death decisions in place
of the usual resampling step. To circumvent the sort of real-time
pitfalls discussed in this work, a random schedule of work units is
used. This requires a central scheduler, so it is not immediately
obvious how the particle cascade might be distributed. In contrast,
we propose, in this work, an SMC algorithm with the conventional structure,
but including parameter estimation as in SMC$^{2}$~\citep{Chopin2013},
and an anytime treatment of the move step. This facilitates distributed
implementation, but provides only indirect control over the total
compute budget.

The construction of Monte Carlo estimators under real-time budgets
has been considered before. \citet*{Heidelberger1988}, \citet*{Glynn1990}
and \citet*{Glynn1991} suggest a number of estimators for the mean
of a random variable after performing independent and identically
distributed (iid) simulation for some prescribed length of real time.
Bias and variance results are established for each. The validity of
their results relies on the exchangeability of simulations conditioned
on their number, and does not extend to MCMC algorithms except for
the special cases of regenerative Markov chains and perfect simulation.
The present work establishes results for MCMC algorithms (for which,
of course, iid sampling is a special case) albeit for a different
problem definition.

A number of other recent works are relevant. Recent papers have considered
the distributed implementation of Gibbs sampling and the implications
of asynchronous updates in this context, which involves real time
considerations~\citep{Terenin2015,DeSa2016}. As mentioned above,
optimisation algorithms already exhibit anytime behaviour and it is
natural to consider whether they might be leveraged to develop anytime
Monte Carlo algorithms, perhaps in a manner similar to the weighted
likelihood bootstrap~\citep{Newton1994}. \citet{Meeds2015} suggest
an approach along this vein for approximate Bayesian computation.
Beyond Monte Carlo, other methods for probabilistic inference might
be considered in an anytime setting. Embedded systems, with organic
real-time constraints, yield natural examples \citep[e.g.][]{Ramos2005}.

The remainder of the paper is structured as follows. Section \ref{sec:framework}
formalises the problem and the framework of a proposed solution; proofs
are deferred to Appendix \ref{sec:proofs}. Section \ref{sec:methods}
uses the framework to establish SMC algorithms with anytime moves,
amongst them an algorithm suitable for large-scale distributed computing.
Section \ref{sec:experiments} validates the framework on a simple
toy problem, and demonstrates the SMC algorithms on a large-scale
distributed computing case study. Section \ref{sec:discussion} discusses
some of the finer points of the results, before Section \ref{sec:conclusion}
concludes.

\section{Framework\label{sec:framework}}

Let $(X_{n})_{n=0}^{\infty}$ be a Markov chain with initial state
$X_{0}$, evolving on a space $\mathbb{X}$, with transition kernel
$X_{n}\mid x_{n-1}\sim\kappa(\mathrm{d}x\,|\,x_{n-1})$ and target
(invariant) distribution $\pi(\mathrm{d}x)$. We do not assume that
$\kappa$ has a density, e.g. it may be a Metropolis\textemdash Hastings
kernel. (In contrast the notation for a probability density would be $\kappa(x\,|\,x_{n-1}).$) A computer processor takes some random and positive real
time $H_{n-1}$ to complete the computations necessary to transition
from $X_{n-1}$ to $X_{n}$ via $\kappa$. Let $H_{n-1}\mid x_{n-1}\sim\tau(\mathrm{d}h\mid x_{n-1})$
for some probability distribution $\tau$. The cumulative distribution
function (cdf) corresponding to $\tau$ is denoted $F_{\tau}(h\mid x_{n-1})$,
and its survival function $\bar{F}_{\tau}(h\mid x_{n-1})=1-F_{\tau}(h\mid x_{n-1})$.
We assume (1) that $H>\epsilon>0$ for some minimal time $\epsilon$,
(2) that $\sup_{x\in\mathbb{X}}\mathbb{E}_{\tau}[H\mid x]<\infty$,
and (3) that $\tau$ is homogeneous in time.

The distribution $\tau$ captures the dependency of hold time on the
current state of the Markov chain, as well as on exogenous factors
such as contesting jobs that run on the same processor, memory management,
I/O, network latency, etc. The first two assumptions seem reasonable:
on a computer processor, a computation must take at least one clock
cycle, justifying a lower bound on $H$, and be expected to complete,
justifying finite expectation. The third assumption, of homogeneity,
is more restrictive. Exogenous factors may produce transient effects,
such as a contesting process that begins part way through the computation
of interest. We discuss the relaxation of this assumption\textemdash as
future work\textemdash in Section \ref{sec:discussion}.

Besides these assumptions, no particular form is imposed on $\tau$.
Importantly, we do not assume that $\tau$ is memoryless. In general,
nothing is known about $\tau$ except how to simulate it, precisely
by recording the length of time $H_{n-1}$ taken to simulate $X_{n}\,|\,x_{n-1}\sim\kappa(\mathrm{d}x\,|\,x_{n-1})$.
In this sense there exists a joint kernel $\kappa(\mathrm{d}x_{n},\mathrm{d}h_{n-1}\mid x_{n-1})=\kappa(\mathrm{d}x_{n}\mid h_{n-1},x_{n-1})\tau(\mathrm{d}h_{n-1}\mid x_{n-1})$
for which the original kernel $\kappa$ is the marginal over all possible
hold times $H_{n-1}$ in transit from $X_{n-1}$ to $X_{n}$:
\[
\kappa(\mathrm{d}x\mid x_{n-1})=\int_{0}^{\infty}\kappa(\mathrm{d}x\mid x_{n-1},h_{n-1})\tau(\mathrm{d}h_{n-1}\mid x_{n-1}).
\]

We now construct a real-time semi-Markov jump process to describe
the progress of the computation in real time. The states of the process
are given by the sequence $(X_{n})_{n=0}^{\infty}$, with associated
hold times $(H_{n})_{n=0}^{\infty}$. Define the arrival time of the
$n$th state $X_n$ as 
\[
A_{n}:=\sum_{i=0}^{n-1}H_{i},
\]
for all $n\geq1$, and $a_{0}=0$, and a process counting the number
of arrivals by time $t$ as
\[
N(t):=\max \left\{ n:A_{n}\leq t\right\} .
\]
Figure \ref{fig:process} illustrates a realisation of the process.

\begin{figure}
\begin{centering}
\includegraphics[width=1\textwidth]{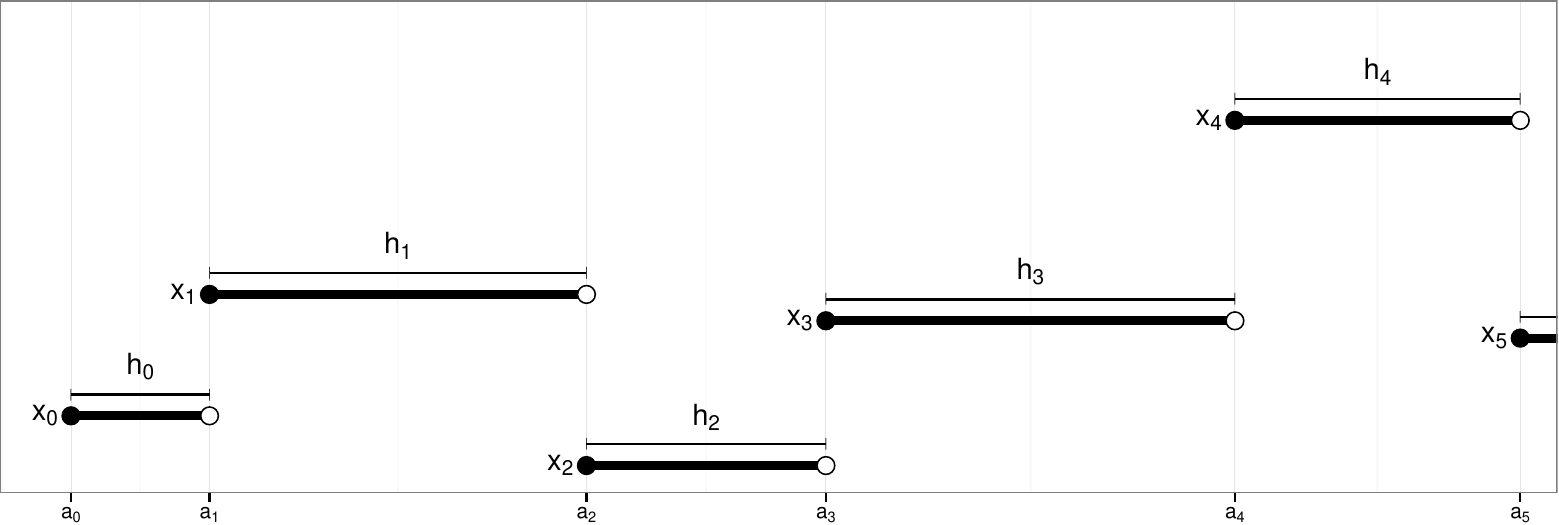} 
\par\end{centering}
\caption{A realisation of a Markov chain $(X_{n})_{n=0}^{\infty}$ in real
time, with hold times $(H_{n})_{n=0}^{\infty}$ and arrival times
$(A_{n})_{n=0}^{\infty}$.\label{fig:process}}
\end{figure}

Now, construct a continuous-time Markov jump process to chart the
progress of the simulation in real time. Let $X(t):=X_{N(t)}$ (the
interpolated process in Figure \ref{fig:process}) and define the
lag time elapsed since the last jump as $L(t):=t-A_{N(t)}$. Then
$(X,L)(t)$ is a Markov jump process.

The process $(X,L)(t)$ is readily manipulated to establish the following
properties. Proofs are given in Appendix \ref{sec:proofs}. 
\begin{prop}
\label{prop:invariant}The stationary distribution of the Markov jump
process $(X,L)(t)$ is 
\[
\alpha(\mathrm{d}x,\mathrm{d}l)=\frac{\bar{F_{\tau}}(l\mid x)}{\mathbb{E}_{\tau}[H]}\pi(\mathrm{d}x)\,\mathrm{d}l.
\]
\end{prop}

\begin{cor}
\label{prop:length-bias}With respect to $\pi(\mathrm{d}x)$, $\alpha(\mathrm{d}x)$
is length-biased by expected hold time: 
\[
\alpha(\mathrm{d}x)=\frac{\mathbb{E}_{\tau}[H\mid x]}{\mathbb{E}_{\tau}[H]}\pi(\mathrm{d}x).
\]
\end{cor}

The interpretation of these results are as follows: if the chain $(X,L)(t)$ is initialised at some time $t=b<0$ such that it is in stationarity at time $t=0$, then $X(0)$'s distribution is precisely given in Corollary \ref{prop:length-bias}. By Corollary \ref{prop:length-bias}, in stationarity, the likelihood
of a particular $x$ appearing is proportional to the likelihood $\pi$ with
which it arises under the original Markov chain, and the expected
length of real time for which it holds when it does.

\begin{cor}
\label{prop:anytime-to-target}Let $(X,L)\sim\alpha$, then the conditional
probability distribution of $X$ given $L<\epsilon$ is $\pi$.
\end{cor}

Corollary \ref{prop:anytime-to-target} simply recovers the original
Markov chain from the Markov jump process. This result follows since $\bar{F_{\tau}}(l\mid x) =1$  (of Proposition \ref{prop:invariant}) for all $x$ and $l \leq\epsilon$.

We refer to $\alpha$ as the \emph{anytime distribution}. It is precisely
the stationary distribution of the Markov jump process. The new name
is introduced to distinguish it from the stationary distribution $\pi$ of
the original Markov chain, which we continue to refer to as the \emph{target
distribution}.

Finally, we state an ergodic theorem from \citet{Alsmeyer1997}; see
also \citet{Alsmeyer1994}. Rather than study the process $(X,L)(t)$,
we can equivalently study $(X_{n},A_{n})_{n=1}^{\infty}$, with the
initial state being $(X_{0},0)$. This is a Markov renewal process.
Conditioned on $(X_{n})_{n=0}^{\infty}$, the hold times $(H_{n})_{n=0}^{\infty}$
are independent. This conditional independence can be exploited to
derive ergodic properties of $(X_{n},A_{n})_{n=1}^{\infty}$, based
on assumed regularity of the driving chain $(X_{n})_{n=0}^{\infty}$.
\begin{prop}[\citealt{Alsmeyer1997}, Corollary\ 1]
\label{prop:ergodic}Assume that the Markov chain $(X_{n})_{n=1}^{\infty}$
is Harris recurrent. For a function $g:\mathbb{X}\rightarrow\mathbb{R}$
with $\int\left|g(x)\right|\alpha(\mathrm{d}x)<\infty$,
\[
\lim_{t\rightarrow\infty}\mathbb{E}\left[g(X(t))\mid x(0),l(0)\right]=\int g(x)\alpha(\mathrm{d}x)
\]
for $\pi$-almost all $x(0)$ and all $l(0)$.
\end{prop}
A further interpretation of this result confirms that, regardless of how we initialise the chain $(X,L)(t)$  at time $t=0$, at a future time $t=T$ the distribution of $X(T)$ is close to $\alpha$ and approaches $\alpha$ as $T\rightarrow \infty$.

\subsection{Establishing anytime behaviour\label{sec:establishing}}

The above results establish that, when interrupted at real time $t$,
the state of a Monte Carlo computation is distributed according to
the anytime distribution, $\alpha$. We wish to establish situations
in which it is instead distributed according to the target distribution,
$\pi$. This will allow us to draw samples of $\pi$ simply by interrupting
the running process at any time $t$, and will form the basis of anytime
Monte Carlo algorithms.

Recalling Corollary \ref{prop:length-bias}, a sufficient condition
to establish $\alpha(\mathrm{d}x)=\pi(\mathrm{d}x)$ is that $\mathbb{E}_{\tau}[H\mid x]=\mathbb{E}_{\tau}[H]$,
i.e. for the expected hold time to be independent of $X$. For iid
sampling, this is trivially the case: we have $\kappa(\mathrm{d}x\mid x_{n-1})=\pi(\mathrm{d}x)$,
the hold time $H_{n-1}$ for $X_{n-1}$ is the time taken to draw
$X_{n}\sim\pi(\mathrm{d}x)$ and so is independent of $X_{n-1}$,
and $\mathbb{E}_{\tau}[H\mid x]=\mathbb{E}_{\tau}[H]$. 

For non-iid sampling, first consider the following change to the Markov
kernel:
\[
X_{n}\mid x_{n-2}\sim\kappa(\mathrm{d}x\mid x_{n-2}).
\]
That is, each new state $X_{n}$ depends not on the previous state,
$x_{n-1}$, but on one state back again, $x_{n-2}$, so that odd-
and even- numbered states are independent. The hold times of the even-numbered
states are the compute times of the odd-numbered states and vice versa,
so that hold times are independent of states, and the desired property
$\mathbb{E}_{\tau}[H\mid x]=\mathbb{E}_{\tau}[H]$, and so $\alpha(\mathrm{d}x)=\pi(\mathrm{d}x)$,
is achieved. This sampling strategy is illustrated in Figure \ref{fig:two_chain} where the odd chain is chain 1 (superscript 1) and the even chain is chain 2. When querying the sampler at any time $t$, the sampler returns the chain that is \emph{holding} and not being worked, which is chain 2 in the figure. It follows then that the returned sample is distributed according to $\pi$ as desired. In an attempt to increase the efficiency of this procedure for SMC methods in Section \ref{sec:methods}, we extend the idea to $K+1$ chains as follows. The processor works/samples each of the $K+1$ chains in sequential order. At any time $t$, when queried, all but one chain is being worked and the processor returns the states of the $K$ chains that are not being worked on. The returned states will be independent of each other with each having distribution $\pi$.

\begin{figure}
\begin{centering}
\includegraphics[width=1\textwidth]{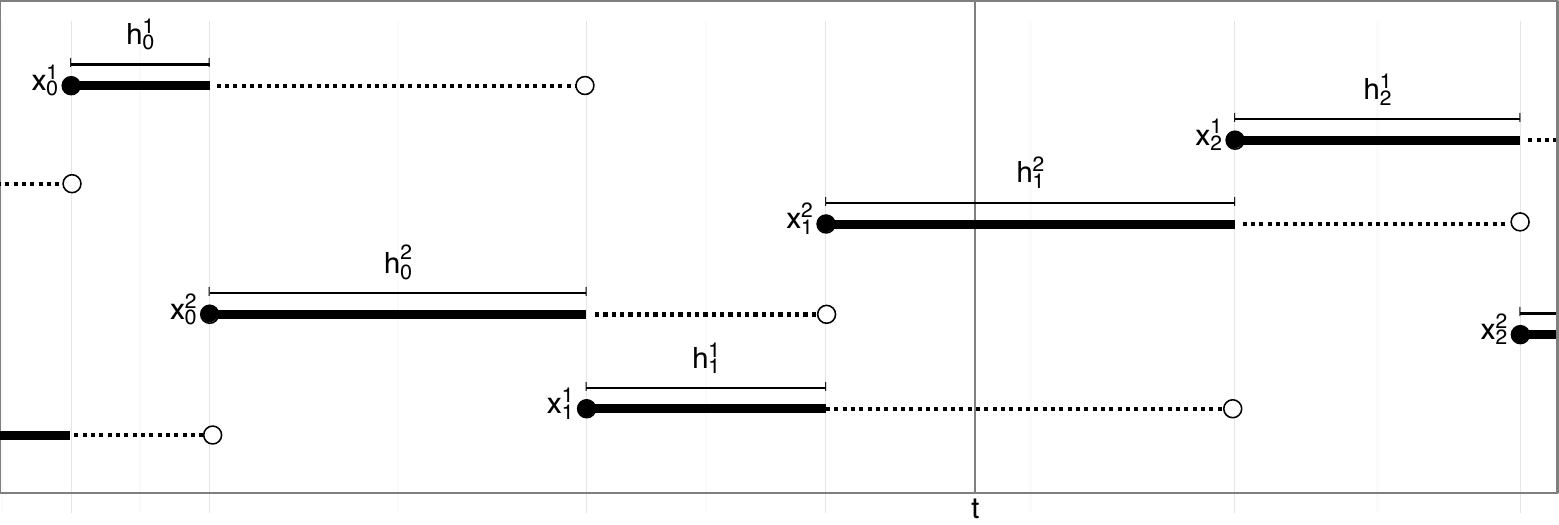} 
\par\end{centering}
\caption{Illustration of the multiple chain concept with two Markov chains.
At any time, one chain is being simulated (indicated with a dotted
line) while one chain is waiting (indicated with a solid line). When
querying the process at some time $t$, it is the state of the waiting
chain that is reported, so that the hold times of each chain are the
compute times of the other chain. For $K+1\geq2$ chains, there is
always one chain simulating while $K$ chains wait, and when querying
the process at some time $t$, the states of all $K$ waiting chains
are reported. \label{fig:two_chain}}
\end{figure}


Formally, suppose that we are simulating $K$ number of Markov chains,
with $K$ a positive integer, plus one extra chain. Denote these $K+1$
chains as $(X_{n}^{1:K+1})_{n=0}^{\infty}$. For simplicity, assume
that all have the same target distribution $\pi$, kernel $\kappa$,
and hold time distribution $\tau$. The joint target is
\[
\Pi(\mathrm{d}x^{1:K+1})=\prod_{k=1}^{K+1}\pi(\mathrm{d}x^{k}).
\]
The $K+1$ chains are simulated on the same processor, one at a time,
in a serial schedule. To avoid introducing an index for the currently
simulating chain, it is equivalent that chain $K+1$ is always the
one simulating, but that states are rotated between chains after each
jump. Specifically, the state of chain $K+1$ at step $n-1$ becomes
the state of chain 1 at step $n$, and the state of each other chain
$k\in\left\{ 1,\ldots,K\right\} $ becomes the state of chain $k+1$.
The transition can then be written as
\[
X_{n}^{1:K+1}\mid x_{n-1}^{1:K+1}\sim\kappa(\mathrm{d}x_{n}^{1}\mid x_{n-1}^{K+1})\prod_{k=1}^{K}\delta_{x_{n-1}^{k}}(\mathrm{d}x_{n}^{k+1}).
\]

As before, this joint Markov chain has an associated joint Markov
jump process $(X^{1:K+1},L)(t)$, where $L(t)$ is the lag time elapsed
since the last jump. This joint Markov jump process is readily manipulated
to yield the following properties, analogous to the single chain case.
Proofs are given in Appendix \ref{sec:proofs}.
\begin{prop}
\label{prop:product-invariant}The stationary distribution of the
Markov jump process $(X^{1:K+1},L)(t)$ is
\begin{equation}
A(\mathrm{d}x^{1:K+1},\mathrm{d}l)=\alpha(\mathrm{d}x^{K+1},\mathrm{d}l)\prod_{k=1}^{K}\pi(\mathrm{d}x^{k}).\label{eq:product-anytime}
\end{equation}
\end{prop}

This result for the joint process extrapolates Proposition \ref{prop:invariant}.
To see this, note that the key ingredients of Proposition \ref{prop:invariant},
$\bar{F}_{\tau}(l\mid x)$ and $\pi(\mathrm{d}x)$, are now replaced
with $\bar{F}_{\tau}(l\mid x^{K+1})$ and $\Pi(\mathrm{d}x^{1:K+1})$.
The hold-time survival function for the product chain is $\bar{F}_{\tau}(l\mid x^{K+1})$,
since it related to the time taken to execute the kernel $\kappa(\mathrm{d}x_{n}^{1}\mid x_{n-1}^{K+1})$,
which depends on the state of chain $K+1$ only.
\begin{cor}
\label{prop:product-length-bias}With respect to $\Pi(\mathrm{d}x^{1:K+1})$,
$A(\mathrm{d}x^{1:K+1})$ is length-biased by expected hold time on
the extra state $X^{K+1}$ only:
\[
A(\mathrm{d}x^{1:K+1})=\alpha(\mathrm{d}x^{K+1})\prod_{k=1}^{K}\pi(\mathrm{d}x^{k}).
\]
\end{cor}

\begin{cor}
\label{prop:product-anytime-to-target}Let $(X^{1:K+1},L)\sim A$,
then the conditional probability distribution of $X^{1:K+1}$ given
$L<\epsilon$ is $\Pi(\mathrm{d}x^{1:K+1})$.
\end{cor}

We state without proof that the product chain construction also satisfies
the ergodic theorem under the same assumptions as Proposition\ \ref{prop:ergodic}.
Numerical validation is given in Section \ref{sec:toy-experiments}.

The practical implication of these properties is that any MCMC algorithm
running $K\geq1$ chains can be converted into an anytime MCMC algorithm
by interleaving one extra chain. When the computation is interrupted
at some time $t$, the state of the extra chain is distributed according
to $\alpha$, while the states of the remaining $K$ chains are independently
distributed according to $\pi$. The state of the extra chain is simply
discarded to eliminate the length bias.

\section{Methods\label{sec:methods}}

It is straightforward to apply the above framework to design anytime
MCMC algorithms. One simply runs $K+1$ chains of the desired MCMC
algorithm on a single processor, using a serial schedule, and eliminating
the state of the extra chain whenever the computation is interrupted
at some real time $t$. The anytime framework is particularly useful
within a broader SMC method~\citep{DelMoral2006}, where there already
exist multiple chains (particles) with which to establish anytime
behaviour. We propose appropriate SMC methods in this section.

\subsection{SMC}

We are interested in the context of sequential Bayesian inference
targeting a posterior distribution $\pi(\mathrm{d}x)=p(\mathrm{d}x\mid y_{1:V})$
for a given data set $y_{1:V}$. For the purposes of SMC, we assume
that the target distribution $\pi(\mathrm{d}x$) admits a density
$\pi(x)$ in order to compute importance weights.

Define a sequence of target distributions $\pi_{0}(\mathrm{d}x)=p(\mathrm{d}x)$
and $\pi_{v}(\mathrm{d}x)\propto p(\mathrm{d}x\mid y_{1:v})$ for
$v=1,\ldots,V$. The first target equals the prior distribution, while
the final target equals the posterior distribution, $\pi_{V}(\mathrm{d}x)=p(\mathrm{d}x\mid y_{1:V})=\pi(\mathrm{d}x)$.
Each target $\pi_{v}$ has an associated Markov kernel $\kappa_{v}$,
invariant to that target, which could be defined using an MCMC algorithm.

An SMC algorithm propagates a set of weighted samples (particles)
through the sequence of target distributions. At step $v$, the target
$\pi_{v}(\mathrm{d}x)$ is represented by an empirical approximation
$\hat{\pi}_{v}(\mathrm{d}x)$, constructed with $K$ number of samples
$x_{v}^{1:K}$ and their associated weights $w_{v}^{1:K}$: 
\begin{equation}
\hat{\pi}_{v}(\mathrm{d}x)=\frac{\sum_{k=1}^{K}w_{v}^{k}\delta_{x_{v}^{k}}(\mathrm{d}x)}{\sum_{k=1}^{K}w_{v}^{k}}.\label{eq:empirical-approx}
\end{equation}
A basic SMC algorithm proceeds as in Algorithm \ref{alg:conventional-smc}.

\begin{algorithm}[tph]
\begin{enumerate}
\item Initialise $x_{0}^{k}\sim\pi_{0}(\mathrm{d}x_{0})$.
\item For $v=1,\ldots,V$ 
\begin{enumerate}
\item Set $x_{v}^{k}=x_{v-1}^{k}$ and weight $w_{v}^{k}=\pi_{v}(x_{v}^{k})/\pi_{v-1}(x_{v}^{k})\propto p(y_{v}\mid x_{v}^{k},y_{1:v-1})$,
to form the empirical approximation $\hat{\pi}_{v}(\mathrm{d}x_{v})\approx\pi_{v}(\mathrm{d}x_{v})$.
\item Resample $x_{v}^{k}\sim\hat{\pi}_{v}(\mathrm{d}x_{v})$.
\item Move $x_{v}^{k}\sim\kappa_{v}(\mathrm{d}x_{v}\mid x_{v}^{k})$ for
$n_{v}$ steps. 
\end{enumerate}
\end{enumerate}
\caption{A basic SMC algorithm. Where $k$ appears, the operation is performed
for all $k\in\{1,\ldots,K\}$.\label{alg:conventional-smc}}
\end{algorithm}

An extension of the algorithm concerns the case where the sequence
of target distributions requires marginalising over some additional
latent variables. In these cases, a pseudomarginal approach~\citep{Andrieu2009}
can be adopted, replacing the exact weight computations with unbiased
estimates~\citep{Fearnhead2008,Fearnhead2010a}. For example, for
a state-space model at step $v$, there are $v$ hidden states $z_{1:v}\in\mathbb{Z}^{v}$
to marginalise over: 
\[
\pi_{v}(\mathrm{d}x)=\int_{\mathbb{Z}^{v}}\pi_{v}(\mathrm{d}x,\mathrm{d}z_{1:v}).
\]
 Unbiased estimates of this integral can be obtained by a nested SMC
procedure targeting $\pi_{v}(\mathrm{d}z_{1:v}\mid x)$, leading to
the algorithm known as SMC$^{2}$~\citep{Chopin2013}, where the
kernels $\kappa_{v}$ are particle MCMC moves~\citep{Andrieu2010}.
This is used for the example of Section \ref{sec:lorenz96-experiments}.

\subsection{SMC with anytime moves}

In the conventional SMC algorithm, it is necessary to choose $n_{v}$,
the number of kernel moves to make per particle in step 2(c). For
anytime moves, this is replaced with a real-time budget $t_{v}$.
Move steps are typically the most expensive steps\textemdash certainly
so for SMC$^{2}$\textemdash with the potential for significant variability
in the time taken to move each particle. An anytime treatment provides
control over the budget of the move step which, if it is indeed the
most expensive step, provides substantial control over the total budget
also.

The anytime framework is used as follows. Associated with each target
distribution $\pi_{v}$, and its kernel $\kappa_{v}$, is a hold time
distribution $\tau_{v}$, and implied anytime distribution $\alpha_{v}$.
At step $v$, after resampling, an extra particle and lag $(x_{v}^{K+1},l_{v})$
are drawn (approximately) from the anytime distribution $\alpha_{v}$.
The real-time Markov jump process $(X_{v}^{1:K+1},L_{v})(t)$ is then
initialised with these particles and lag, and simulated forward until
time $t_{v}$ is reached. The extra chain and lag are then eliminated,
and the states of the remaining chains are restored as the $K$ particles
$x_{v}^{1:K}$.

The complete algorithm is given in Algorithm \ref{alg:anytime-smc}.

\begin{algorithm}[tph]
\begin{enumerate}
\item Initialise $x_{0}^{k}\sim\pi_{0}(\mathrm{d}x_{0})$.
\item For $v=1,\ldots,V$ 
\begin{enumerate}
\item Set $x_{v}^{k}=x_{v-1}^{k}$ and weight $w_{v}^{k}=\pi_{v}(x_{v}^{k})/\pi_{v-1}(x_{v}^{k})\propto p(y_{v}\mid x_{v}^{k},y_{1:v-1})$,
to form the empirical approximation $\hat{\pi}_{v}(\mathrm{d}x_{v})\approx\pi_{v}(\mathrm{d}x_{v})$.
\item Resample $x_{v}^{k}\sim\hat{\pi}_{v}(\mathrm{d}x_{v})$.
\item Draw (approximately) an extra particle and lag $(x_{v}^{K+1},l_{v})\sim\alpha_{v}(\mathrm{d}x_{v},\mathrm{d}l_{v})$.
Construct the real-time Markov jump process $\left(X_{v}^{1:K+1},L_{v}\right)(0)=(x_{v}^{1:K+1},l_{v})$
and simulate it forward for some real time $t_{v}$. Set $x_{v}^{1:K}=X_{v}^{1:K}(t_{v})$,
discarding the extra particle and lag.
\end{enumerate}
\end{enumerate}
\caption{SMC with anytime moves. Where $k$ appears, the operation is performed
for all $k\in\{1,\ldots,K\}$.\label{alg:anytime-smc}}
\end{algorithm}

By the end of the move step 2(c), as $t_{v}\rightarrow\infty$, the
particles $x_{v}^{1:K+1}$ become distributed according to $A_{v}$,
regardless of their distribution after the resampling step 2(b). This
is assured by Proposition \ref{prop:ergodic}. After eliminating the
extra particle, the remaining $x_{v}^{1:K}$ are distributed according
to $\Pi_{v}$.

In practice, of course, it is necessary to choose some finite $t_{v}$
for which the $x_{v}^{1:K}$ are distributed only approximately according
to $\Pi_{v}$. For any given $t_{v}$, their divergence in distribution
from $\Pi_{v}$ is minimised by an initialisation as close as possible
to $A_{v}$. We have, already, the first $K$ chains initialised from
an empirical approximation of the target, $\hat{\pi}_{v}$, which
is unlikely to be improved upon. We need only consider the extra particle
and lag.

An easily-implemented choice is to draw $(x_{v}^{K+1},l_{v})\sim\hat{\pi}_{v}(\mathrm{d}x_{v})\delta_{0}(\mathrm{d}l_{v})$.
In practice, this merely involves resampling $K+1$ rather than $K$
particles in step 2(b), setting $l_{v}=0$ and proceeding with the
first move.

An alternative is $(x_{v}^{K+1},l_{v})\sim\delta_{x_{v-1}^{K+1}}(\mathrm{d}x_{v})\delta_{l_{v-1}}(\mathrm{d}l_{v})$.
This resumes the computation of the extra particle that was discarded
at step $v-1$. As $t_{v-1}\rightarrow\infty$, it amounts to approximating
$\alpha_{v}$ by $\alpha_{v-1}$, which is sensible if the sequence
of anytime distributions changes only slowly. 

\subsection{Distributed SMC with anytime moves}

While the potential to parallelise SMC is widely recognised~\citep[see e.g.][]{Lee2010,Murray2015},
the resampling step 2(b) in Algorithm \ref{alg:conventional-smc}
is acknowledged as a potential bottleneck when in a distributed computing
environment of $P$ number of processors. This is due to collective
communication: all processors must synchronise after completing the
preceding steps in order for resampling to proceed. Resampling cannot
proceed until the slowest among them completes. As this is a maximum
among $P$ processors, the expected wait time increases with $P$.
Recent work has considered either global pairwise interaction~\citep{Murray2011a,Murray2016}
or limited interaction~\citep{Verge2015,Whiteley2016,Lee2016} to
address this issue. Instead, we propose to preserve collective communication,
but to use an anytime move step to ensure the simultaneous readiness
of all processors for resampling.

SMC with anytime moves is readily distributed across multiple processors.
The $K$ particles are partitioned so that processor $p\in\left\{ 1,\ldots,P\right\} $
has some number of particles, denoted $K^{p}$, and so that $\sum_{p=1}^{P}K^{p}=K$.
Each processor can proceed with initialisation, move and weight steps
independently of the other processors. After the resampling step,
each processor has $K^{p}$ number of particles. During the move step,
each processor draws its own extra particle and lag from the anytime
distribution, giving it $K^{p}+1$ particles, and discards them at
the end of the step leaving it with $K^{p}$ again. Collective communication
is required for the resampling step, and an appropriate distributed
resampling scheme should be used~\citep[see e.g.][]{Bolic2005,Verge2015,Lee2016}. 

In the simplest case, all workers have homogeneous hardware and the
obvious partition of particles is $K^{p}=K/P$. For heterogeneous
hardware another partition may be set \emph{a priori}~\citep[see e.g.][]{Rosenthal2000}.
Note also that with heterogenous hardware, each processor may have
a different compute capability and therefore different distribution
$\tau_{v}$. For processor $p$, we denote this $\tau_{v}^{p}$ and
the associated anytime distribution $\alpha_{v}^{p}$. This difference
between processors is easily accommodated, as the anytime treatment
is local to each processor.

A distributed SMC algorithm with anytime moves proceeds as in Algorithm
\ref{alg:distributed-anytime-smc}.

\begin{algorithm}[tph]
\begin{enumerate}
\item On each processor $p$, initialise $x_{0}^{k}\sim\pi_{0}(\mathrm{d}x)$.
\item For $v=1,\ldots,V$ 
\begin{enumerate}
\item On each processor $p$, set $x_{v}^{k}=x_{v-1}^{k}$ and weight $w_{v}^{k}=\pi_{v}(x_{v}^{k})/\pi_{v-1}(x_{v}^{k})\propto p(y_{v}\mid x_{v}^{k},y_{1:v-1})$.
Collectively, all $K$ particles form the empirical approximation
$\hat{\pi}_{v}(\mathrm{d}x_{v})\approx\pi_{v}(\mathrm{d}x_{v})$.
\item Collectively resample $x_{v}^{k}\sim\hat{\pi}_{v}(\mathrm{d}x_{v})$
and redistribute the particles among processors so that processor
$p$ has exactly $K^{p}$ particles again.
\item On each processor $p$, draw (approximately) an extra particle and
lag $(x_{v}^{K^{p}+1},l_{v})\sim\alpha_{v}^{p}(\mathrm{d}x,\mathrm{d}l)$.
Construct the real-time process $\left(X_{v}^{1:K^{p}+1},L_{v}\right)(0)=(x_{v}^{1:K^{p}+1},l_{v})$
and simulate it forward for some real time $t_{v}$. Set $x_{v}^{1:K^{p}}=X_{v}^{1:K^{p}}(t_{v})$,
discarding the extra particle and lag.
\end{enumerate}
\end{enumerate}
\caption{SMC with anytime moves. Where $k$ appears, the operation is performed
for all $k\in\{1,\ldots,K^{p}\}$.\label{alg:distributed-anytime-smc}}
\end{algorithm}

The preceding discussion around the approximate anytime distribution
still holds \emph{for each processor in isolation}: for any given
budget $t_{v}$, to minimise the divergence between the distribution
of particles and the target distribution, $\hat{A}_{v}^{p}$ should
be chosen as close as possible to $A_{v}^{p}$.

\subsection{Setting the compute budget}

We set an overall compute budget for move steps, which we denote $t$,
and apportion this into a quota for each move step $v$, which we
denote $t_{v}$ as above. This requires some \emph{a priori} knowledge
of the compute profile for the problem at hand.

Given $\hat{\pi}_{v-1}$, if the compute time necessary to obtain
$\hat{\pi}_{v}$ is constant with respect to $v$, then a suitable
quota for the $v$th move step is the obvious $t_{v}=t/V$. If, instead,
the compute time grows linearly in $v$, as is the case for SMC$^{2}$,
then we expect the time taken to complete the $v$th step to be proportional
to $v+c$ (where the constant $c$ is used to capture overheads).
A sensible quota for the $v$th move step is then 
\begin{equation}
t_{v}=\left(\frac{v+c}{\sum_{u=1}^{V}(u+c)}\right)t=\left(\frac{2(v+c)}{V(V+2c+1)}\right)t.\label{eq:apportion}
\end{equation}
For the constant, a default choice might be $c=0$; higher values
shift more time to earlier time steps.

This approximation does neglect some complexities. The use of memory-efficient
path storage~\citep{Jacob2015}, for example, introduces a time-dependent
contribution of $\mathcal{O}(K)$ at $v=1$, increasing to $\mathcal{O}(K\log K)$
with $v$ as the ancestry tree grows. Nonetheless, for the example
of Section \ref{sec:lorenz96-experiments} we observe, anecdotally,
that this partitioning of the time budget produces surprisingly consistent
results with respect to the random number of moves completed at each
move step $v$.

\subsection{Resampling considerations}

To reduce the variance in resampling outcomes~\citep{Douc2005},
implementations of SMC often use schemes such as systematic, stratified~\citep{Kitagawa1996}
or residual~\citep{Liu1998} resampling, rather than the multinomial
scheme~\citep{Gordon1993} with which the above algorithms have been
introduced. The implementation of these alternative schemes does not
necessarily leave the random variables $X_{v}^{1},\ldots,X_{v}^{K}$
exchangeable; for example, the offspring of a particle are typically
neighbours in the output of systematic or stratified resampling~\citep[see e.g.][]{Murray2016}.

Likewise, distributed resampling schemes do not necessarily redistribute
particles identically between processors. For example, the implementation
in LibBi~\citep{Murray2015} attempts to minimise the transport of
particles between processors, such that the offspring of a parent
particle are more likely to remain on the same processor as that particle.
This means that the distribution of the $K^{p}$ particles on each
processor may have greater divergence from $\pi_{v}$ than the distribution
of the $K$ particles overall.

In both cases, the effect is that particles are initialised further
from the ideal $A_{v}$. Proposition \ref{prop:ergodic} nonetheless
ensures consistency as $t_{v}\rightarrow\infty$. A random permutation
of particles may result in a better initialisation, but this can be
costly, especially in a distributed implementation where particles
must be transported between processors. For a fixed total budget,
the time spent permuting may be better spent by increasing $t_{v}$.
The optimal choice is difficult to identify in general; we return
to this point in the discussion.

\section{Experiments\label{sec:experiments}}

This section presents two case studies to empirically investigate
the anytime framework and the proposed SMC methods. The first uses
a simple model where real-time behaviour is simulated in order to
validate the results of Section \ref{sec:framework}. The second considers
a Lorenz '96 state-space model with non-trivial likelihood and compute
profile, testing the SMC methods of Section \ref{sec:methods} in
two real-world computing environments.

\subsection{Simulation study\label{sec:toy-experiments}}

Consider the model 
\begin{eqnarray*}
X & \sim & \mathrm{Gamma}(k,\theta)\\
H\mid x & \sim & \mathrm{Gamma}(x^{p}/\theta,\theta),
\end{eqnarray*}
with shape parameter $k$, scale parameter $\theta$, and polynomial
degree $p$. The two Gamma distributions correspond to $\pi$ and
$\tau$, respectively. The anytime distribution is:
\[
\alpha(\mathrm{d}x)=\frac{\mathbb{E}_{\tau}[H\mid x]}{\mathbb{E}_{\tau}[H]}\pi(\mathrm{d}x)\propto x^{k+p-1}\exp\left(-\frac{x}{\theta}\right)\mathrm{d}x,
\]
which is $\mathrm{Gamma}(k+p,\theta)$.

Of course, in real situations, $\tau$ is not known explicitly, and
is merely implied by the algorithm used to simulate $X$. For this
first study, however, we assume the explicit form above and simulate
virtual hold times. This permits exploration of the real-time effects
of polynomial computational complexity in a controlled environment,
including constant ($p=0$), linear ($p=1$), quadratic ($p=2$) and
cubic ($p=3$) complexity.

To construct a Markov chain $(X_{n})_{n=0}^{\infty}$ with target
distribution $\mathrm{Gamma}(k,\theta)$, first consider a Markov
chain $(Z_{n})_{n=0}^{\infty}$ with target distribution $\mathcal{N}(0,1)$
and kernel 
\[
Z_{n}\mid z_{n-1}\sim\mathcal{N}(\rho z_{n-1},1-\rho^{2}),
\]
where $\rho$ is an autocorrelation parameter. Now define $(X_{n})_{n=0}^{\infty}$
as 
\[
x_{n}=F_{\gamma}^{-1}\left(F_{\phi}(z_{n});\,k,\theta\right),
\]
where $F_{\gamma}^{-1}$ is the inverse cdf of the Gamma distribution
with parameters $k$ and $\theta$, and $F_{\phi}$ is the cdf of
the standard normal distribution. By construction, $\rho$ parameterises
a Gaussian copula inducing correlation between adjacent elements of
$(X_{n})_{n=0}^{\infty}$.

For the experiments in this section, we set $k=2$, $\theta=1/2$,
$\rho=1/2$ and use $p\in\left\{ 0,1,2,3\right\} $. In all cases
Markov chains are initialised from $\pi$ and simulated for 200 units
of virtual time.

We employ the 1-Wasserstein distance to compare distributions. For
two univariate distributions $\mu$ and $\nu$ with associated cdfs
$F_{\mu}(x)$ and $F_{\nu}(x)$, the 1-Wasserstein distance $d_{1}(F_{\mu},F_{\nu})$
can be evaluated as~\citep[p.64]{Shorack1986}
\[
d_{1}(F_{\mu},F_{\nu})=\int_{-\infty}^{\infty}\left|F_{\mu}(x)-F_{\nu}(x)\right|\,\mathrm{d}x,
\]
which, for the purposes of this example, is sufficiently approximated
by numerical integration. The first distribution will always be the
empirical distribution of a set of $n$ samples, its cdf denoted $F_{n}(x)$.
If those samples are distributed according to the second distribution,
the distance will go to zero as $n$ increases.

\subsubsection{Validation of the anytime distribution}

We first validate, empirically, that the anytime distribution is indeed
$\textrm{Gamma}(k+p,\theta)$ as expected. We simulate $n=2^{18}$
Markov chains. At each integer time we take the state of all $n$
chains to construct an empirical distribution. We then compute the
1-Wasserstein distance between this and the anytime distribution,
using the empirical cdf $F_{n}(x)$ and anytime cdf $F_{\alpha}(x)$.

Figure \ref{fig:toy-anytime} plots the results over time. In all
cases the distance goes to zero in $t$, slower for larger $p$. This
affirms the theoretical results obtained in Section \ref{sec:framework}.

\begin{figure}[tp]
\includegraphics[width=1\textwidth]{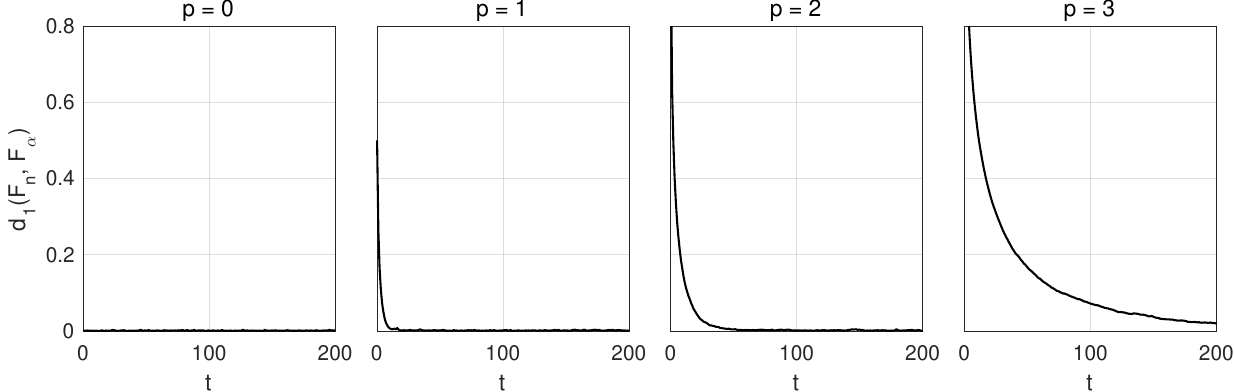}

\caption{Convergence of Markov chains to the anytime distribution for the simulation
study, with constant ($p=0$), linear ($p=1$), quadratic ($p=2$)
and cubic ($p=3$) expected hold time. Each plot shows the evolution
of the 1-Wasserstein distance between the anytime distribution and
the empirical distribution of $2^{18}$ independent Markov chains
initialised from the target distribution.\label{fig:toy-anytime}}
\end{figure}

\subsubsection{Validation of the multiple chain strategy}

We next check the efficacy of the multiple chain strategy in eliminating
length bias. For $K+1\in\left\{ 2,4,8,16,32\right\} $, we initialise
$K+1$ chains and simulate them forward in a serial schedule. For
$n=2^{18}$, this is repeated $n/(K+1)$ times. We then consider ignoring
the length bias versus correcting for it. In the first case, we take
the states of all $K+1$ chains at each time, giving $n$ samples
from which to construct an empirical cdf $F_{n}(x)$. In the second
case, we eliminate the extra chain but keep the remaining $K$, giving
$nK/(K+1)$ samples from which to construct an empirical cdf $F_{nK/(K+1)}(x)$.
In both cases we compute the 1-Wasserstein distance between the empirical
and target distributions, using the appropriate empirical cdf, and
the target cdf $F_{\pi}(x)$.

Figure \ref{fig:toy-target} plots the results over time for both
the uncorrected (top) and corrected (bottom) cases. For the uncorrected
case, the 1-Wasserstein distance between the empirical distribution
and target distribution does not converge to zero. Neither does it
become arbitrarily bad: the distance is due to one of the $K+1$ chains
being distributed according to $\alpha$ and not $\pi$, the influence
of which decreases as $K$ increases.

For the corrected case, where the extra chain is eliminated, the distance
converges to zero in time. This confirms the efficacy of the multiple
chain strategy in yielding an anytime distribution equal to the target
distribution.

\begin{figure}[t]
\includegraphics[width=1\textwidth]{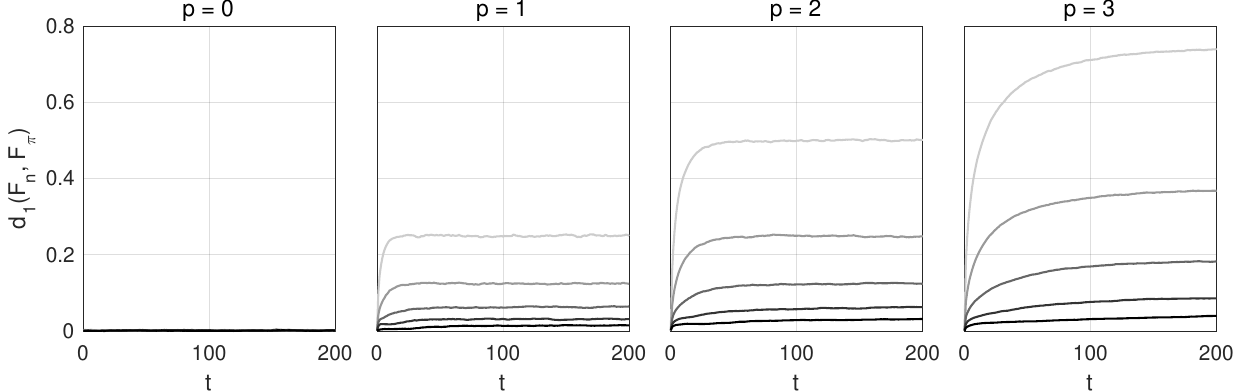}

\includegraphics[width=1\textwidth]{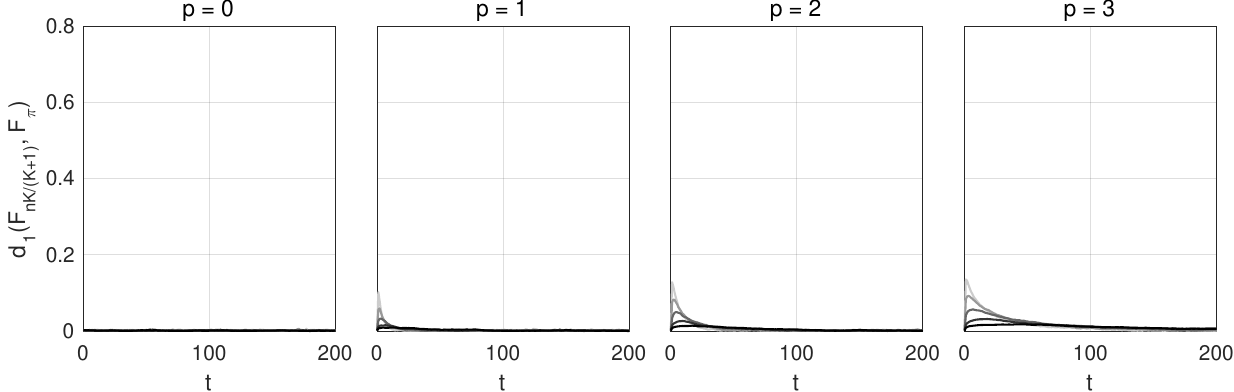}

\caption{Correction of length bias for the simulation study, using $K+1\in\left\{ 2,4,8,16,32\right\} $
chains (light to dark), with constant ($p=0$), linear ($p=1$), quadratic
($p=2$) and cubic ($p=3$) expected hold time. Each plot shows the
evolution of the 1-Wasserstein distance between the empirical and
target distributions. On the top row, the states of all chains contribute
to the empirical distribution, which does not converge to the target.
On the bottom row, the state of the extra chain is eliminated, contributing
only the remaining states to the empirical distribution, which does
converge to the target. \label{fig:toy-target}}
\end{figure}

\subsection{Distributed computing study\label{sec:lorenz96-experiments}}

Consider a stochastic extension of the deterministic Lorenz '96~\citep{Lorenz2006}
model described by the stochastic differential equation (SDE) 
\begin{equation}
\mathrm{d}X_{d}=\left(X_{d-1}(X_{d+1}-X_{d-2})-X_{d}+F\right)\,\mathrm{d}t+\sigma\,\mathrm{d}W_{d},\label{eqn:lorenz96-sde}
\end{equation}
with parameter $F$, constant $\sigma$, state vector $\mathbf{X}(t)\in\mathbb{R}^{D}$
and Wiener process vector $\mathbf{W}(t)\in\mathbb{R}^{D}$, with
elements of those vectors indexed cyclically by subscripts (i.e. $X_{d-D}\equiv X_{d}\equiv X_{d+D}$).
The SDE may be equivalently interpreted in the Ito or Stratonovich
sense, as the noise term is additive~\citep[p157]{Kloeden1992}.
The observation model is given by 
\[
Y_{d}(t)\sim\mathcal{N}(x_{d}(t),\varsigma^{2}).
\]
We fix $D=8$, $\sigma^{2}=10^{-4}$, $\varsigma^{2}=10^{-6}$ and
set a prior on the parameter $F$ and initial conditions $\mathbf{X}(0)$
of: 
\begin{eqnarray*}
F & \sim & \mathcal{U}([0,7])\\
X_{d}(0) & \sim & \mathcal{N}(0,\sigma^{2}).
\end{eqnarray*}

The SDE can be approximately decomposed into a deterministic drift
component given by the ordinary differential equation (ODE) 
\[
\frac{\mathrm{d}x_{d}}{\mathrm{d}t}=x_{d-1}(x_{d+1}-x_{d-2})-x_{d}+F,
\]
and a diffusion component given by the Wiener process. On a fixed
time step $\Delta t=5\times10^{-2}$, the drift component is first
simulated using an appropriate numerical scheme for ODEs. Then, a
Wiener process increment $\Delta W_{d}\sim\mathcal{N}(0,\Delta t)$
is simulated and added to the result. This numerical scheme yields
a result similar to that of Euler\textendash Maruyama for the original
SDE but, for drift, substitutes the usual first-order Euler method
with a higher-order Runge\textendash Kutta method. This is advantageous
in low-noise regimes where $\sigma$ is close to zero, as here. In
such cases the dynamics are drift-dominated and can benefit from the
higher-order scheme \citep[see e.g.][ch. 3]{Milstein2004}.

The RK4(3)5{[}2R+{]}C algorithm of \citet{Kennedy2000} is used to
simulate the drift. This provides a fourth order solution to the ODE
with an embedded third order solution for error estimates. Adaptive
step-size adjustment is then used as in \citet{Hairer1993}. The complete
method is implemented~\citep{Murray2012} on a graphics processing
unit (GPU) as in the LibBi software~(\url{www.libbi.org}, \citealp{Murray2015}).

The Lorenz '96 model exhibits intricate qualitative behaviours that
depend on the parameter $F$. These range from decay, to periodicity,
to chaos and back again (Figure \ref{fig:lorenz96}, top row). With
an adaptive step-size adjustment, the number of steps required to
simulate trajectories within given error bounds generally increases
with $F$, so that compute time does also.

We produce a data set by setting $F=4.8801$, simulating a single
trajectory for 10 time units and taking partial observations $\mathbf{Y}_{1:4}(t)$
every 0.4 time units. This gives 100 observations in total. We then
use SMC$^{2}$ to attempt to recover the correct posterior distribution
over $F$ given this data set. This is non-trivial: this particular
value of $F$ is in a region where the qualitative behaviour of the
Lorenz '96 model appears to switch frequently, in $F$, between periodic
and chaotic regimes (Figure \ref{fig:lorenz96}, top right), suggesting
that the marginal likelihood may not be smooth in the region of the
posterior, and inference may be difficult.

The marginal likelihood $p(y\mid F)$ cannot be computed exactly,
but it can be unbiasedly estimated with SMC. For each value of $F$
on a regular grid, we run SMC with $2^{20}$ particles to estimate
the marginal likelihood. These estimates are shown in the middle left
of Figure \ref{fig:lorenz96}. The likelihood is clearly multi-modal,
and the estimator heteroskedastic. Nevertheless, the variance in the
estimator is tolerable in the region of the posterior distribution
(middle right of Figure \ref{fig:lorenz96}), suggesting that $F$
can be recovered. The real time taken to compute these estimates is
shown in the lower plots of Figure \ref{fig:lorenz96}. The computations
were performed in a random order through the grid points on $F$ so
as to decorrelate $F$ with any transient exogenous effects on compute
time. There appears, in fact, to have been some such effect: note
the dotted line of points above the bulk on each plot, suggesting
that a subset of runs have been slowed. This is most likely due to
a contesting process on the shared server on which these computations
were run. As expected, compute time tends to increase in $F$ (after
an initial plateau where other factors dominate). Furthermore, variance
appears to increase with $F$ in the higher regions.

\begin{figure}[p]
\includegraphics[width=1\textwidth]{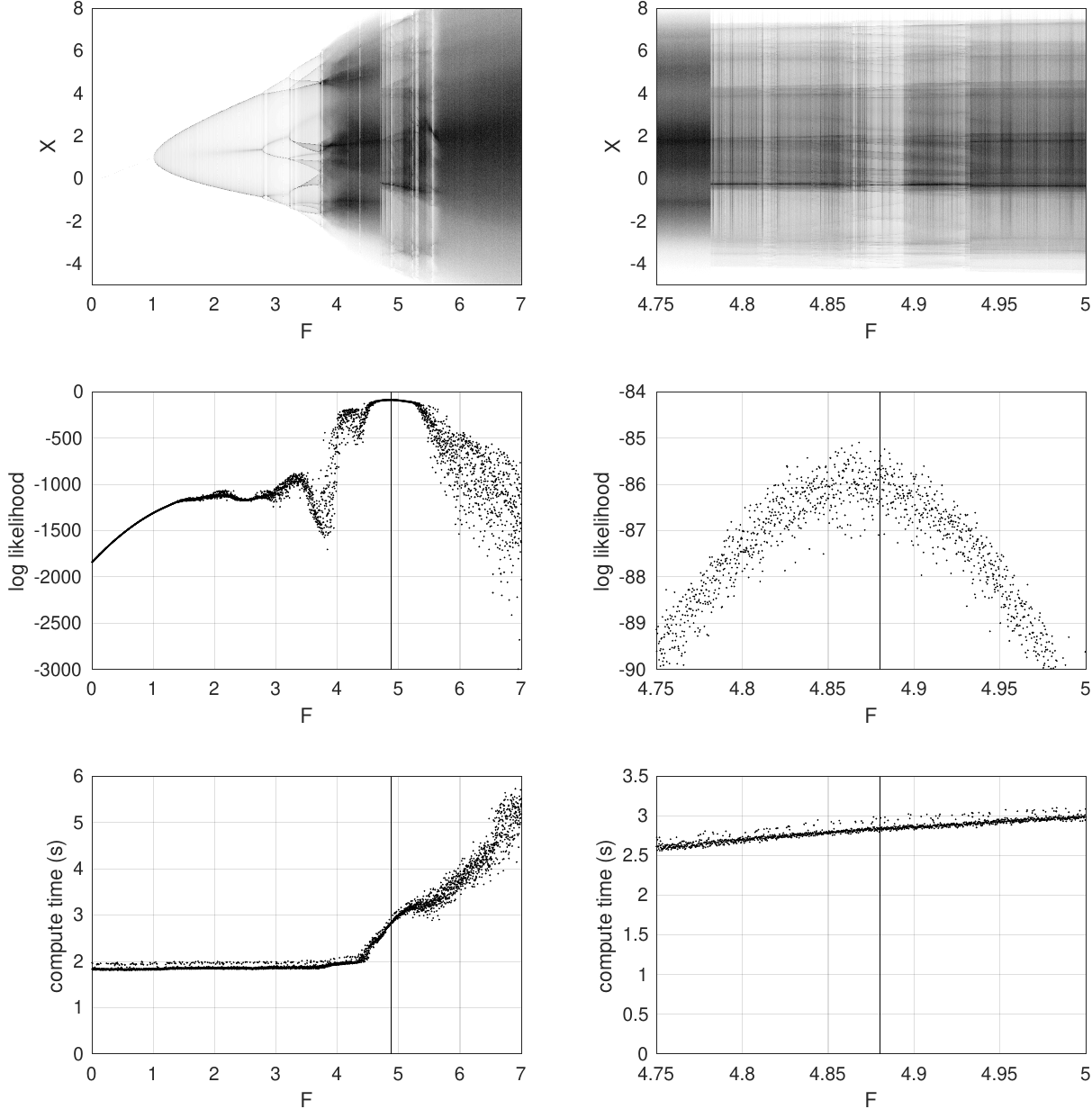}

\caption{Elucidating the Lorenz '96 model. The left column shows the range
$F\in[0,7]$ as in the uniform prior distribution, while the right
column shows a narrower range of $F$ in the vicinity of the posterior distribution.
The solid vertical lines indicate the value $F=4.8801$, with which
data is simulated. The first row is a bifurcation diagram depicting
the stationary distribution of any element of $\mathbf{X}(t)$ for
various values of $F$. Each column is a density plot for a particular
value of $F$; darker for higher density values, scaled so that the
mode is black. Note the intricate behaviours of decay, periodicity
and chaos induced by $F$. The second row depicts estimates of the
marginal log-likelihood of the simulated data set for the same values
of $F$, using SMC with $2^{20}$ particles. Multiple modes and heteroskedasticity
are apparent. The third row depicts the compute time taken to obtain
these estimates, showing increasing compute time in $F$ after an
initial plateau.\label{fig:lorenz96}}
\end{figure}

We now run SMC$^{2}$ using the LibBi software on two platforms: 
\begin{enumerate}
\item A shared-memory machine with 8 GPUs, each with 1536 cores, for approximately
12000-way parallelism, using $2^{10}$ particles for $F$, each with
$2^{20}$ particles for $\mathbf{X}(t)$, for approximately one billion
particles overall. This is a shared machine where contestation from
other jobs is expected.
\item A distributed-memory cluster on the Amazon EC2 service, with 128 GPUs,
each with 1536 cores, for approximately 200000-way parallelism, using
$2^{12}$ particles for $F$, each with $2^{20}$ particles for $\mathbf{X}(t)$,
for approximately four billion particles overall. This is a dedicated
cluster where contestation from other jobs is not expected. 
\end{enumerate}
In order to obtain a more repeatable comparison between conventional
SMC$^{2}$ and SMC$^{2}$ with anytime moves, we choose to use the
same number of samples of $\mathbf{X}(t)$ for all time steps, rather
than adapting this in time as recommended in \citet{Chopin2011}.
For the same reason, we resample at all steps rather than use an adaptive
trigger. With anytime moves, the extra particle and lag are drawn
as $(x_{v}^{K+1},l_{v})\sim\hat{\pi}_{v}(\mathrm{d}x_{v})\delta_{0}(\mathrm{d}l_{v})$.

We first run conventional SMC$^{2}$, making $n_{v}=10$ moves per
particle at each step $v$. We then run SMC$^{2}$ with anytime moves,
prescribing a total budget for move steps of 60 minutes for the 8
GPU configuration, and 5 minutes for the 128 GPU configuration, apportioned
as in (\ref{eq:apportion}).

The results of the 128 GPU runs are given in Figure \ref{fig:smc-posteriors}.
Recalling that $F=4.8801$ for the simulated data set, these suggest
that the posterior has indeed been recovered successfully, and there
is no indication that the posterior obtained with anytime move steps
is much different from that obtained using the conventional method.

Compute profiles for the runs are given in Figure \ref{fig:smc-profiles},
showing the busy and wait times of all processors involved in the
computations. We see obvious wait time with conventional SMC$^{2}$,
far more pronounced in the 8 GPU case, where a contesting process
on one processor has encumbered the entire computation. The anytime
move step grants a robustness to this contesting process, and wait
times are significantly reduced. For the 128 GPU case, even in the
absence of such exogenous problems, wait times are noticeably reduced.

\begin{figure}[p]
\includegraphics[width=0.95\textwidth]{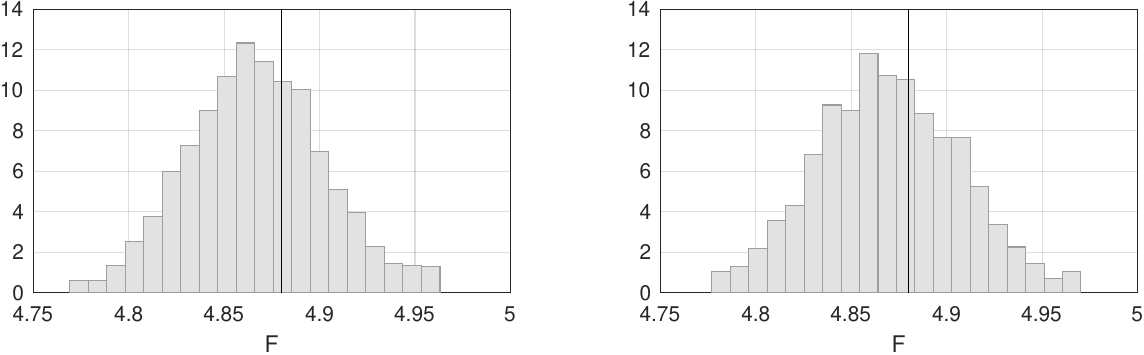}

\caption{Posterior distributions over $F$ for the Lorenz '96 case study. On
the left, from conventional SMC$^{2}$, on the right, from SMC$^{2}$
with anytime moves, both running on the 128 GPU configuration.\label{fig:smc-posteriors}}
\end{figure}

\begin{figure}[p]
\includegraphics[width=0.95\textwidth]{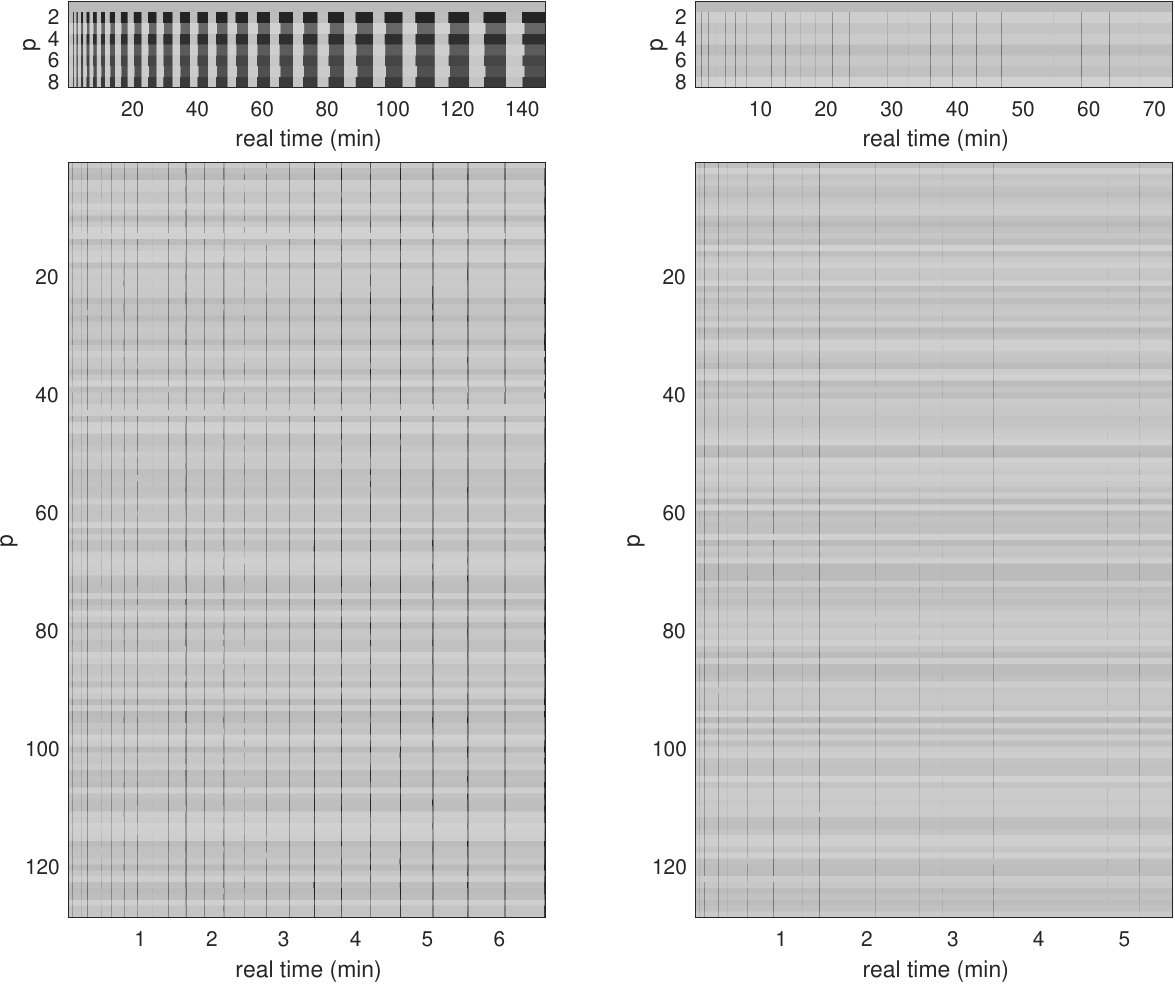}

\caption{Compute profiles for the Lorenz '96 case study. On the left is a conventional
distributed SMC$^{2}$ method with a fixed number of moves per particle
after resampling. On the right is distributed SMC$^{2}$ with anytime
move steps. Each row represents the activity of a single processor
over time: light grey while active and dark grey while waiting. The
top profiles are for an 8 GPU shared system where contesting processes
are expected. The conventional method on the left exhibits significant
idle time on processors 2 to 8 due to a contesting job on processor
1. The two bottom profiles are for the 128 GPU configuration with
no contesting processes. Wait time in the conventional methods on
the left is significantly reduced in the anytime methods on the right. \label{fig:smc-profiles}}
\end{figure}

\section{Discussion\label{sec:discussion}}

The framework presented is a generic means by which any MCMC algorithm\textemdash including
iid sampling as a special case\textemdash can be made an anytime Monte
Carlo algorithm. This facilitates the configuration of Monte Carlo
computation in real-time terms, rather than in the number of simulations.
The benefits of this have been demonstrated in a distributed computing
context where, by setting real-time compute budgets, wait times are
significantly reduced for an SMC algorithm that requires collective
communication. The framework has potential applications elsewhere,
for example as a foundation for real-time, fault-tolerant and energy-constrained
Monte Carlo algorithms, for the management of cloud computing budgets,
or for the fair computational comparison of methods.

We have assumed throughout that an algorithm is given to simulate
the target distribution $\pi$, and that the anytime distribution
$\alpha$ is merely a consequence of this. The aim has then been to
correct the length bias in $\alpha$. This is a pragmatic approach,
as it leverages existing Monte Carlo algorithms. A tantalising alternative
is to develop algorithms that, from the outset, yield $\pi$ as the
anytime distribution. This might be done with an underlying Markov
chain that targets something other than $\pi$ but that, by design,
yields $\pi$ once length biased. We expect, however, that to do this
even approximately will require at least some knowledge of $\tau$,
which will restrict its applicability to specific cases only.

The proposed SMC method uses anytime move steps, but is not a complete
anytime algorithm, as it does not provide control over the total compute
budget. Its objective is to minimise the wait time that precedes resampling
steps in a distributed implementation of SMC. On this account it succeeds.
A complete anytime SMC algorithm (of a conventional structure) will
require, in addition, anytime weighting and anytime resampling steps,
as well as the apportioning of the total compute budget between these.
Because approximations may appear in each of these steps, the apportioning
is not straightforward, and will involve tradeoffs. As already identified,
for example, the redistribution of particles after resampling in a
distributed environment is an expensive operation, and all or part
of that time may be better invested in the budget allocation for anytime
moves. Such investigations have been left to future work. An alternative
means to an anytime SMC algorithm is to use a different structure
to the conventional, as in the particle cascade~\citep{Paige2014}.
Whatever the structure, these anytime algorithms are somewhat more
elaborate than the standard SMC algorithms for which theoretical results
have been established, and may warrant further study.

Finally, we return to the strongest of the assumptions of the whole
framework: that of the homogeneity of $\tau$ in time. This may be
unrealistic in the presence of transient exogenous factors, such as
intermittent network issues, or contesting processes running on the
same hardware only temporarily. If the assumption is relaxed, so that
$\tau$ varies in time, the anytime distribution will vary as well,
and ergodicity will not hold. Figure \ref{fig:toy-target} suggests
that, for example, an exogenous switching factor in $\tau$ would
induce transient effects in the anytime distribution that are not
necessarily eliminated by the multiple chain strategy. There may be
weaker assumptions under which comparable results and appropriate
methods can be established, but this investigation is left to future
work.

\section{Conclusion\label{sec:conclusion}}

This work has presented an approach to allow any MCMC algorithm to
be made an anytime Monte Carlo algorithm, eliminating the length bias
associated with real-time budgets. This is particularly important
in situations where the final state of a Markov chain is more important
than computing averages over all states. It has applications in embedded,
distributed, cloud, real-time and fault-tolerant computing. To demonstrate
the usefulness of the approach, a new SMC$^{2}$ method has been presented,
which exhibits significantly reduced wait time when run on a large-scale
distributed computing system.

\section*{Disclosure Statements}

\emph{Data Availability:} The methods introduced in this paper are implemented in LibBi version
1.3.0, and the empirical results may be reproduced with the LibBi
\emph{Anytime} package. Both are available from \url{www.libbi.org}.

\section*{Acknowledgements}

The authors would like to thank the Isaac Newton Institute for Mathematical
Sciences, Cambridge, for support and hospitality during the programme
\emph{Monte Carlo Methods for Complex Inference Problems} (MCMW01),
where some of this work was undertaken. This work was also financially
supported by an EPSRC-Cambridge Big Data Travel Grant and EPSRC (EP/K020153/1),
The Alan Turing Institute under the EPSRC grant EP/N510129/1, and
the Swedish Foundation for Strategic Research (SSF) via the project
\emph{ASSEMBLE}. The authors would like to thank Pierre Jacob for helpful conversations.

\bibliographystyle{abbrvnat}
\bibliography{anytime}

\begin{thebibliography}{40}
\providecommand{\natexlab}[1]{#1}
\providecommand{\url}[1]{\texttt{#1}}
\expandafter\ifx\csname urlstyle\endcsname\relax
  \providecommand{\doi}[1]{doi: #1}\else
  \providecommand{\doi}{doi: \begingroup \urlstyle{rm}\Url}\fi

\bibitem[Alsmeyer(1994)]{Alsmeyer1994}
G.~Alsmeyer.
\newblock On the {M}arkov renewal theorem.
\newblock \emph{Stochastic Processes and their Applications}, 50:\penalty0
  37--56, 1994.

\bibitem[Alsmeyer(1997)]{Alsmeyer1997}
G.~Alsmeyer.
\newblock The {M}arkov renewal theorem and related results.
\newblock \emph{Markov Processes and Related Fields}, 3:\penalty0 103--127,
  1997.

\bibitem[Andrieu and Roberts(2009)]{Andrieu2009}
C.~Andrieu and G.~O. Roberts.
\newblock The pseudo-marginal approach for efficient {M}onte {C}arlo
  computations.
\newblock \emph{Annals of Statistics}, 37\penalty0 (2):\penalty0 697--725, 04
  2009.
\newblock \doi{10.1214/07-AOS574}.

\bibitem[Andrieu et~al.(2010)Andrieu, Doucet, and Holenstein]{Andrieu2010}
C.~Andrieu, A.~Doucet, and R.~Holenstein.
\newblock Particle {M}arkov chain {M}onte {C}arlo methods.
\newblock \emph{Journal of the Royal Statistical Society B}, 72:\penalty0
  269--302, 2010.
\newblock \doi{10.1111/j.1467-9868.2009.00736.x}.

\bibitem[Boli\'{c} et~al.(2005)Boli\'{c}, Djuri\'{c}, and Hong]{Bolic2005}
M.~Boli\'{c}, P.~M. Djuri\'{c}, and S.~Hong.
\newblock Resampling algorithms and architectures for distributed particle
  filters.
\newblock \emph{IEEE Transactions on Signal Processing}, 53:\penalty0
  2442--2450, 2005.
\newblock \doi{10.1109/TSP.2005.849185}.

\bibitem[Chopin et~al.(2011)Chopin, Jacob, and Papaspiliopoulos]{Chopin2011}
N.~Chopin, P.~Jacob, and O.~Papaspiliopoulos.
\newblock {SMC}$^2$: A sequential {M}onte {C}arlo algorithm with particle
  {M}arkov chain {M}onte {C}arlo updates.
\newblock 2011.

\bibitem[Chopin et~al.(2013)Chopin, Jacob, and Papaspiliopoulos]{Chopin2013}
N.~Chopin, P.~Jacob, and O.~Papaspiliopoulos.
\newblock {SMC}$^2$: An efficient algorithm for sequential analysis of state
  space models.
\newblock \emph{Journal of the Royal Statistical Society B}, 75:\penalty0
  397--426, 2013.
\newblock \doi{10.1111/j.1467-9868.2012.01046.x}.

\bibitem[d'Avigneau et~al.(2020)d'Avigneau, Singh, and Murray]{ASM20}
A.~d'Avigneau, S.~Singh, and L.~Murray.
\newblock Anytime parallel tempering.
\newblock \emph{arXiv.org e-Print archive}, 2020.

\bibitem[{De Sa} et~al.(2016){De Sa}, Olukotun, and Ré]{DeSa2016}
C.~{De Sa}, K.~Olukotun, and C.~Ré.
\newblock Ensuring rapid mixing and low bias for asynchronous {G}ibbs sampling,
  2016.

\bibitem[{Del Moral} et~al.(2006){Del Moral}, Doucet, and Jasra]{DelMoral2006}
P.~{Del Moral}, A.~Doucet, and A.~Jasra.
\newblock Sequential {M}onte {C}arlo samplers.
\newblock \emph{Journal of the Royal Statistical Society B}, 68:\penalty0
  441--436, 2006.
\newblock \doi{10.1111/j.1467-9868.2006.00553.x}.

\bibitem[Douc and Capp\'{e}(2005)]{Douc2005}
R.~Douc and O.~Capp\'{e}.
\newblock Comparison of resampling schemes for particle filtering.
\newblock In \emph{Image and Signal Processing and Analysis, 2005. ISPA 2005.
  Proceedings of the 4th International Symposium on}, pages 64 -- 69, 15-17
  2005.

\bibitem[Fearnhead et~al.(2008)Fearnhead, Papaspiliopoulos, and
  Roberts]{Fearnhead2008}
P.~Fearnhead, O.~Papaspiliopoulos, and G.~O. Roberts.
\newblock Particle filters for partially observed diffusions.
\newblock \emph{Journal of the Royal Statistical Society B}, 70:\penalty0
  755--777, 2008.
\newblock \doi{10.1111/j.1467-9868.2008.00661.x}.

\bibitem[Fearnhead et~al.(2010)Fearnhead, Papaspiliopoulos, Roberts, and
  Stuart]{Fearnhead2010a}
P.~Fearnhead, O.~Papaspiliopoulos, G.~O. Roberts, and A.~Stuart.
\newblock Random-weight particle filtering of continuous time processes.
\newblock \emph{Journal of the Royal Statistical Society B}, 72\penalty0
  (4):\penalty0 497--512, 2010.
\newblock \doi{10.1111/j.1467-9868.2010.00744.x}.

\bibitem[Glynn and Heidelberger(1990)]{Glynn1990}
P.~W. Glynn and P.~Heidelberger.
\newblock Bias properties of budget constraint simulations.
\newblock \emph{Operations Research}, 38\penalty0 (5):\penalty0 801--814, 1990.

\bibitem[Glynn and Heidelberger(1991)]{Glynn1991}
P.~W. Glynn and P.~Heidelberger.
\newblock Analysis of parallel replicated simulations under a completion time
  constraint.
\newblock \emph{ACM Transactions on Modeling and Computer Simulations},
  1\penalty0 (1):\penalty0 3--23, 1991.

\bibitem[Gordon et~al.(1993)Gordon, Salmond, and Smith]{Gordon1993}
N.~Gordon, D.~Salmond, and A.~Smith.
\newblock Novel approach to nonlinear/non-{G}aussian {B}ayesian state
  estimation.
\newblock \emph{IEE Proceedings-F}, 140:\penalty0 107--113, 1993.
\newblock \doi{10.1049/ip-f-2.1993.0015}.

\bibitem[Hairer et~al.(1993)Hairer, {N\o rsett}, and Wanner]{Hairer1993}
E.~Hairer, S.~{N\o rsett}, and G.~Wanner.
\newblock \emph{Solving Ordinary Differential Equations I: Nonstiff Problems}.
\newblock Springer--Verlag, 2 edition, 1993.

\bibitem[Heidelberger(1988)]{Heidelberger1988}
P.~Heidelberger.
\newblock Discrete event simulations and parallel processing: Statistical
  properties.
\newblock \emph{SIAM Journal on Scientific and Statistical Computing},
  9\penalty0 (6):\penalty0 1114--1132, 1988.
\newblock \doi{10.1137/0909077}.

\bibitem[Jacob et~al.(2015)Jacob, Murray, and Rubenthaler]{Jacob2015}
P.~E. Jacob, L.~M. Murray, and S.~Rubenthaler.
\newblock Path storage in the particle filter.
\newblock \emph{Statistics and Computing}, 25\penalty0 (2):\penalty0 487--496,
  2015.
\newblock ISSN 0960-3174.
\newblock \doi{10.1007/s11222-013-9445-x}.

\bibitem[Kennedy et~al.(2000)Kennedy, Carpenter, and Lewis]{Kennedy2000}
C.~A. Kennedy, M.~H. Carpenter, and R.~M. Lewis.
\newblock Low-storage, explicit {R}unge--{K}utta schemes for the compressible
  {N}avier-{S}tokes equations.
\newblock \emph{Applied Numerical Mathematics}, 35:\penalty0 177--219, 2000.

\bibitem[Kitagawa(1996)]{Kitagawa1996}
G.~Kitagawa.
\newblock Monte {C}arlo filter and smoother for non-{G}aussian nonlinear state
  space models.
\newblock \emph{Journal of Computational and Graphical Statistics}, 5:\penalty0
  1--25, 1996.
\newblock \doi{10.2307/1390750}.

\bibitem[Kloeden and Platen(1992)]{Kloeden1992}
P.~E. Kloeden and E.~Platen.
\newblock \emph{Numerical Solution of Stochastic Differential Equations}.
\newblock Springer--Verlag, 1992.

\bibitem[Lee and Whiteley(2016)]{Lee2016}
A.~Lee and N.~Whiteley.
\newblock Forest resampling for distributed sequential {M}onte {C}arlo.
\newblock \emph{Statistical Analysis and Data Mining: The ASA Data Science
  Journal}, 9\penalty0 (4):\penalty0 230--248, 2016.
\newblock ISSN 1932-1872.
\newblock \doi{10.1002/sam.11280}.

\bibitem[Lee et~al.(2010)Lee, Yau, Giles, Doucet, and Holmes]{Lee2010}
A.~Lee, C.~Yau, M.~B. Giles, A.~Doucet, and C.~C. Holmes.
\newblock On the utility of graphics cards to perform massively parallel
  simulation of advanced {M}onte {C}arlo methods.
\newblock \emph{Journal of Computational and Graphical Statistics},
  19:\penalty0 769--789, 2010.
\newblock \doi{10.1198/jcgs.2010.10039}.

\bibitem[Liu and Chen(1998)]{Liu1998}
J.~S. Liu and R.~Chen.
\newblock Sequential {M}onte-{C}arlo methods for dynamic systems.
\newblock \emph{Journal of the American Statistical Association}, 93:\penalty0
  1032--1044, 1998.

\bibitem[Lorenz(2006)]{Lorenz2006}
E.~N. Lorenz.
\newblock \emph{Predictability of Weather and Climate}, chapter 3:
  Predictability -- A Problem Partly Solved, pages 40--58.
\newblock Cambridge University Press, 2006.
\newblock \doi{10.1017/CBO9780511617652.004}.

\bibitem[Meeds and Welling(2015)]{Meeds2015}
T.~Meeds and M.~Welling.
\newblock Optimization {M}onte {C}arlo: Efficient and embarrassingly parallel
  likelihood-free inference.
\newblock In C.~Cortes, N.~D. Lawrence, D.~D. Lee, M.~Sugiyama, and R.~Garnett,
  editors, \emph{Advances in Neural Information Processing Systems 28}, pages
  2080--2088. Curran Associates, Inc., 2015.

\bibitem[Milstein and Tretyakov(2004)]{Milstein2004}
G.~N. Milstein and M.~V. Tretyakov.
\newblock \emph{Stochastic Numerics for Mathematical Physics}.
\newblock Scientific Computation. Springer--Verlag, 2004.
\newblock \doi{10.1007/978-3-662-10063-9}.

\bibitem[Murray(2011)]{Murray2011a}
L.~M. Murray.
\newblock {GPU} acceleration of the particle filter: The {M}etropolis
  resampler.
\newblock In \emph{DMMD: Distributed machine learning and sparse representation
  with massive data sets}, 2011.
\newblock URL \url{http://arxiv.org/abs/1202.6163}.

\bibitem[Murray(2012)]{Murray2012}
L.~M. Murray.
\newblock {GPU} acceleration of {R}unge--{K}utta integrators.
\newblock \emph{IEEE Transactions on Parallel and Distributed Systems},
  23:\penalty0 94--101, 2012.
\newblock \doi{10.1109/TPDS.2011.61}.

\bibitem[Murray(2015)]{Murray2015}
L.~M. Murray.
\newblock Bayesian state-space modelling on high-performance hardware using
  {LibBi}.
\newblock \emph{Journal of Statistical Software}, 67\penalty0 (10):\penalty0
  1--36, 2015.
\newblock ISSN 1548-7660.
\newblock \doi{10.18637/jss.v067.i10}.

\bibitem[Murray et~al.(2016)Murray, Lee, and Jacob]{Murray2016}
L.~M. Murray, A.~Lee, and P.~E. Jacob.
\newblock Parallel resampling in the particle filter.
\newblock \emph{Journal of Computational and Graphical Statistics}, 25\penalty0
  (3):\penalty0 789--805, 2016.
\newblock \doi{10.1080/10618600.2015.1062015}.

\bibitem[Newton and Raftery(1994)]{Newton1994}
M.~A. Newton and A.~E. Raftery.
\newblock Approximate {B}ayesian inference with the weighted likelihood
  bootstrap.
\newblock \emph{Journal of the Royal Statistical Society B}, 56\penalty0
  (1):\penalty0 3--48, 1994.
\newblock ISSN 00359246.

\bibitem[Paige et~al.(2014)Paige, Wood, Doucet, and Teh]{Paige2014}
B.~Paige, F.~Wood, A.~Doucet, and Y.~W. Teh.
\newblock Asynchronous anytime sequential {M}onte {C}arlo.
\newblock In \emph{Advances in Neural Information Processing Systems 27}, pages
  3410--3418. 2014.

\bibitem[Ramos and Cozman(2005)]{Ramos2005}
F.~T. Ramos and F.~G. Cozman.
\newblock Anytime anyspace probabilistic inference.
\newblock \emph{International Journal of Approximate Reasoning}, 38\penalty0
  (1):\penalty0 53 -- 80, 2005.
\newblock ISSN 0888-613X.
\newblock \doi{10.1016/j.ijar.2004.04.001}.

\bibitem[Rosenthal(2000)]{Rosenthal2000}
J.~S. Rosenthal.
\newblock Parallel computing and {M}onte {C}arlo algorithms.
\newblock \emph{Far East Journal of Theoretical Statistics}, 4:\penalty0
  207--236, 2000.

\bibitem[Shorack and Wellner(1986)]{Shorack1986}
G.~R. Shorack and J.~A. Wellner.
\newblock \emph{Empirical processes with applications to statistics}.
\newblock Wiley, New York, 1986.
\newblock ISBN 9780471867258;047186725X;.

\bibitem[Terenin et~al.(2015)Terenin, Simpson, and Draper]{Terenin2015}
A.~Terenin, D.~Simpson, and D.~Draper.
\newblock Asynchronous {G}ibbs sampling.
\newblock 2015.
\newblock URL \url{https://arxiv.org/abs/1509.08999}.

\bibitem[Verg{\'e} et~al.(2015)Verg{\'e}, Dubarry, Del~Moral, and
  Moulines]{Verge2015}
C.~Verg{\'e}, C.~Dubarry, P.~Del~Moral, and E.~Moulines.
\newblock On parallel implementation of sequential {M}onte {C}arlo methods: the
  island particle model.
\newblock \emph{Statistics and Computing}, 25\penalty0 (2):\penalty0 243--260,
  2015.
\newblock ISSN 1573-1375.
\newblock \doi{10.1007/s11222-013-9429-x}.

\bibitem[Whiteley et~al.(2016)Whiteley, Lee, and Heine]{Whiteley2016}
N.~Whiteley, A.~Lee, and K.~Heine.
\newblock On the role of interaction in sequential {M}onte {C}arlo algorithms.
\newblock \emph{Bernoulli}, 22\penalty0 (1):\penalty0 494--529, 2016.
\newblock \doi{10.3150/14-BEJ666}.

\end{thebibliography}

\appendix

\section{Proofs\label{sec:proofs}}

Recall the assumptions that $H>\epsilon>0$ for some minimal time
$\epsilon$, that $\sup_{x\in\mathbb{X}}\mathbb{E}_{\tau}[H\mid x]<\infty$,
and that $\tau$ is homogeneous in time.

For any positive $\Delta\leq\epsilon$, define the notation: 
\begin{align*}
x & :=x(t), & l & :=l(t), & x_{+} & :=x(t+\Delta), & l_{+} & :=l(t+\Delta).
\end{align*}

At most one jump can occur in any time interval $(t,t+\Delta]$, and
it may occur at any time in that interval. This implies that either
$l_{+}=l+\Delta-h$ with $l_{+}\in[0,\Delta)$ if a jump occurs, or
$l_{+}=l+\Delta$ if no jump occurs.

To simplify the proof of Proposition \ref{prop:invariant} somewhat,
we further assume that $\tau$ admits a density, although this is
not strictly necessary. The below proofs are similar if one wishes
to adopt discrete hold times; assuming, for example, that jumps can
only occur at the end of a processor's clock cycle. Here we assume
that jumps can occur at any real time, and that we may query the state
of the process at any real time.
\begin{proof}[Proof of Proposition \ref{prop:invariant}]
The transition kernel is
\[
\lambda(\mathrm{d}x_{+},\mathrm{d}l_{+}\mid x,l)=\rho(x,l)\lambda_{1}(\mathrm{d}x_{+},\mathrm{d}l_{+}\mid x,l)+(1-\rho(x,l))\lambda_{0}(\mathrm{d}x_{+},\mathrm{d}l_{+}\mid x,l),
\]
where 
\begin{align*}
\rho(x,l) & =\frac{F_{\tau}(l+\Delta\mid x)-F_{\tau}(l\mid x)}{\bar{F}_{\tau}(l\mid x)\hfill} &  & \text{is the probability that a jump occurs,}\\
\lambda_{1}(\mathrm{d}x_{+},\mathrm{d}l_{+}\mid x,l) & =\frac{\kappa(\mathrm{d}x_{+}\mid x,l+\Delta-l_{+})\tau(l+\Delta-l_{+}\mid x)\mathbb{I}_{[0,\Delta)}(\mathrm{d}l_{+})}{F_{\tau}(l+\Delta\mid x)-F_{\tau}(l\mid x)\hfill} &  & \text{is the transition if a jump occurs, and}\\
\lambda_{0}(\mathrm{d}x_{+},\mathrm{d}l_{+}\mid x,l) & =\delta_{x}(\mathrm{d}x_{+})\delta_{l+\Delta}(\mathrm{d}l_{+}) &  & \text{is the transition if no jump occurs}.
\end{align*}

We have:
\begin{align*}
 & \int_{\mathbb{X}}\int_{0}^{\infty}\lambda(\mathrm{d}x_{+},\mathrm{d}l_{+}\mid x,l)\alpha(\mathrm{d}x,\mathrm{d}l)\\
= & \underbrace{\int_{\mathbb{X}}\int_{0}^{\infty}\rho(x,l)\lambda_{1}(\mathrm{d}x_{+},\mathrm{d}l_{+}\mid x,l)\alpha(\mathrm{d}x,\mathrm{d}l)}_{\lyxmathsym{\ding{202}}}+\underbrace{\int_{\mathbb{X}}\int_{0}^{\infty}\left(1-\rho(x,l)\right)\lambda_{0}(\mathrm{d}x_{+},\mathrm{d}l_{+}\mid x,l)\alpha(\mathrm{d}x,\mathrm{d}l)}_{\lyxmathsym{\ding{203}}}.
\end{align*}
To demonstrate that $\alpha(\mathrm{d}x,\mathrm{d}l)$ is the invariant
distribution, it is sufficient to show that the above equals $\alpha(\mathrm{d}x_{+},\mathrm{d}l_{+})$:
\begin{flalign*}
\text{\ding{202}} & =\frac{\mathbb{I}_{[0,\Delta)}(\mathrm{d}l_{+})}{\mathbb{E}_{\tau}[H]}\int_{\mathbb{X}}\int_{0}^{\infty}\kappa(\mathrm{d}x_{+}\mid x,l+\Delta-l_{+})\tau(l+\Delta-l_{+}\mid x)\,\mathrm{d}l\,\pi(\mathrm{d}x)\\
 & =\frac{\mathbb{I}_{[0,\Delta)}(\mathrm{d}l_{+})}{\mathbb{E}_{\tau}[H]}\int_{\mathbb{X}}\int_{0}^{\infty}\kappa(\mathrm{d}x_{+}\mid x,l)\tau(l\mid x)\,\mathrm{d}l\,\pi(\mathrm{d}x) & \text{(since \ensuremath{\Delta-l_{+}\in(0,\Delta]} and \ensuremath{F_{\tau}(\Delta\mid x)=0)}}\\
 & =\frac{\mathbb{I}_{[0,\Delta)}(\mathrm{d}l_{+})}{\mathbb{E}_{\tau}[H]}\int_{\mathbb{X}}\kappa(\mathrm{d}x_{+}\mid x)\pi(\mathrm{d}x)\\
 & =\frac{\mathbb{I}_{[0,\Delta)}(\mathrm{d}l_{+})}{\mathbb{E}_{\tau}[H]}\pi(\mathrm{d}x_{+})\\
 & =\mathbb{I}_{[0,\Delta)}(l_{+})\alpha(\mathrm{d}x_{+},\mathrm{d}l_{+}), & \text{(since \ensuremath{\bar{F}_{\tau}(l_{+}\mid x_{+})=1} for \ensuremath{l_{+}\in[0,\Delta)})}\\
\text{\ding{203}} & =\mathbb{I}_{[\Delta,\infty)}(\mathrm{d}l_{+})\left(1-\rho(x_{+},l_{+}-\Delta)\right)\alpha(\mathrm{d}x_{+},\mathrm{d}l_{+}-\Delta)\\
 & =\mathbb{I}_{[\Delta,\infty)}(\mathrm{d}l_{+})\left(\frac{\bar{F}_{\tau}(l_{+}\mid x)}{\bar{F}_{\tau}(l_{+}-\Delta\mid x)}\right)\left(\frac{\bar{F}_{\tau}(l_{+}-\Delta\mid x)}{\mathbb{E}_{\tau}[H]}\right)\pi(\mathrm{d}x_{+})\\
 & =\mathbb{I}_{[\Delta,\infty)}(\mathrm{d}l_{+})\alpha(\mathrm{d}x_{+},\mathrm{d}l_{+}),\\
\text{\ding{202}}+\text{\ding{203}} & =\alpha(\mathrm{d}x_{+},\mathrm{d}l_{+}). & \qedhere
\end{flalign*}
\end{proof}
\begin{proof}[Proof of Corollary \ref{prop:length-bias}]
We have:
\[
\alpha(\mathrm{d}x)=\frac{\int_{0}^{\infty}\bar{F}_{\tau}(l\mid x)\,\mathrm{d}l}{\mathbb{E}_{\tau}[H]}\pi(\mathrm{d}x)=\frac{\mathbb{E}_{\tau}[H\mid x]}{\mathbb{E}_{\tau}[H]}\pi(\mathrm{d}x).\qedhere
\]
\end{proof}
\begin{proof}[Proof of Corollary \ref{prop:anytime-to-target}]
We have:
\begin{flalign*}
\alpha(\mathrm{d}x,L<\Delta) & =\frac{\int_{0}^{\Delta}\bar{F}_{\tau}(l\mid x)\,\mathrm{d}l}{\mathbb{E}_{\tau}[H]}\pi(\mathrm{d}x)=\frac{\Delta}{\mathbb{E}_{\tau}[H]}\pi(\mathrm{d}x) & \text{(since \ensuremath{\bar{F}_{\tau}(l\mid x)=1} for \ensuremath{l\in[0,\Delta)})}\\
\alpha(L<\Delta) & =\frac{\Delta}{\mathbb{E}_{\tau}[H]}\\
\alpha(\mathrm{d}x\mid L<\Delta) & =\frac{\alpha(\mathrm{d}x,L<\Delta)}{\alpha(L<\Delta)}=\pi(\mathrm{d}x). & \qedhere
\end{flalign*}
\end{proof}
\begin{proof}[Proof of Proposition \ref{prop:product-invariant}]
The transition kernel is:
\[
\nu(\mathrm{d}x_{+}^{1:K+1},\mathrm{d}l_{+}\,|\,x^{1:K+1},l)=\rho(x^{K+1},l)\nu_{1}(\mathrm{d}x_{+}^{1:K+1},\mathrm{d}l_{+}\,|\,x^{1:K+1},l)+\left(1-\rho(x^{K+1},l)\right)\nu_{0}(\mathrm{d}x_{+}^{1:K+1},\mathrm{d}l_{+}\mid x^{1:K+1},l),
\]
where
\begin{align*}
\nu_{1}(\mathrm{d}x_{+}^{1:K+1},\mathrm{d}l_{+}\,|\,x^{1:K+1},l) & =\lambda_{1}(\mathrm{d}x_{+}^{1},\mathrm{d}l_{+}\mid x^{K+1},l)\prod_{k=1}^{K}\delta_{x^{k}}(\mathrm{d}x_{+}^{k+1}) &  & \text{is the transition if a jump occurs, and}\\
\nu_{0}(\mathrm{d}x_{+}^{1:K+1},\mathrm{d}l_{+}\,|\,x^{1:K+1},l) & =\lambda_{0}(\mathrm{d}x_{+}^{K+1},\mathrm{d}l_{+}\mid x^{K+1},l)\prod_{k=1}^{K}\delta_{x^{k}}(\mathrm{d}x_{+}^{k}) &  & \text{is the transition if a jump does not occur}.
\end{align*}

We have:
\begin{align*}
 & \int_{\mathbb{X}^{K+1}}\int_{0}^{\infty}\nu(\mathrm{d}x_{+}^{1:K+1},\mathrm{d}l_{+}\mid x^{1:K+1},l)A(\mathrm{d}x^{1:K+1},\mathrm{d}l)\\
= & \underbrace{\int_{\mathbb{X}^{K+1}}\int_{0}^{\infty}\rho(x^{K+1},l)\nu_{1}(\mathrm{d}x_{+}^{1:K+1},\mathrm{d}l_{+}\mid x^{1:K+1},l)A(\mathrm{d}x^{1:K+1},\mathrm{d}l)}_{\text{\ding{202}}}\\
+ & \underbrace{\int_{\mathbb{X}^{K+1}}\int_{0}^{\infty}\left(1-\rho(x^{K+1},l)\right)\nu_{0}(\mathrm{d}x_{+}^{1:K+1},\mathrm{d}l_{+}\mid x^{1:K+1},l)A(\mathrm{d}x^{1:K+1},\mathrm{d}l)}_{\text{\ding{203}}},
\end{align*}
To demonstrate that $A(\mathrm{d}x^{1:K+1},\mathrm{d}l)$ is the invariant
distribution, it is sufficient to show that the above equals $A(\mathrm{d}x_{+}^{1:K+1},\mathrm{d}l_{+})$:
\begin{flalign*}
\text{\ding{202}} & =\int_{\mathbb{X}}\int_{0}^{\infty}\lambda_{1}(\mathrm{d}x_{+}^{1},\mathrm{d}l_{+}\mid x^{K+1},l)\alpha(\mathrm{d}x^{K+1},\mathrm{d}l)\prod_{k=1}^{K}\int_{\mathbb{X}}\delta_{x^{k}}(\mathrm{d}x_{+}^{k+1})\pi(\mathrm{d}x^{k})\\
 & =\mathbb{I}_{[0,\Delta)}(l_{+})\alpha(\mathrm{d}x_{+}^{1},\mathrm{d}l_{+})\prod_{k=1}^{K}\pi(\mathrm{d}x_{+}^{k+1})\\
 & =\mathbb{I}_{[0,\Delta)}(l_{+})\alpha(\mathrm{d}x_{+}^{K+1},\mathrm{d}l_{+})\prod_{k=1}^{K}\pi(\mathrm{d}x_{+}^{k}) & \text{(by Corollary 3)}\\
 & =\mathbb{I}_{[0,\Delta)}(l_{+})A(\mathrm{d}x_{+}^{1:K+1},\mathrm{d}l_{+}),\\
\text{\ding{203}} & =\mathbb{I}_{[\Delta,\infty)}(l_{+})\int_{0}^{\infty}\lambda_{0}(\mathrm{d}x_{+}^{K+1},\mathrm{d}l_{+}\mid x^{K+1},l)\alpha(\mathrm{d}x^{K+1},\mathrm{d}l)\prod_{k=1}^{K}\int_{\mathbb{X}}\delta_{x^{k}}(\mathrm{d}x_{+}^{k})\pi(\mathrm{d}x^{k})\\
 & =\mathbb{I}_{[\Delta,\infty)}(l_{+})\alpha(\mathrm{d}x_{+}^{K+1},\mathrm{d}l_{+})\prod_{k=1}^{K}\pi(\mathrm{d}x_{+}^{k})\\
 & =\mathbb{I}_{[\Delta,\infty)}(l_{+})A(\mathrm{d}x_{+}^{1:K+1},\mathrm{d}l_{+}),\\
\text{\ding{202}}+\text{\ding{203}} & =A(\mathrm{d}x_{+}^{1:K+1},\mathrm{d}l_{+}). & \qedhere
\end{flalign*}
\end{proof}

\end{document}